\documentclass[a4paper,english,10pt]{article}

\usepackage[utf8]{inputenc}
\usepackage{amsmath,amsfonts,amsthm,amssymb,stmaryrd}
\usepackage{enumerate,hyperref}
\usepackage{algorithm}
\usepackage{algorithmic}
\usepackage{verbatim}
\usepackage[textwidth=14.5cm,textheight=23.5cm]{geometry}
\usepackage{tikz}
\usetikzlibrary{calc,backgrounds,arrows,positioning,decorations.pathmorphing,shapes}

\theoremstyle{plain}\newtheorem{theorem}{Theorem}
\theoremstyle{plain}
\theoremstyle{plain}\newtheorem{proposition}[theorem]{Proposition}
\theoremstyle{plain}\newtheorem{lemma}[theorem]{Lemma}
\theoremstyle{plain}\newtheorem{definition}{Definition}
\theoremstyle{plain}\newtheorem{remark}{Remark}
\theoremstyle{plain}\newtheorem{example}[remark]{Example}

\DeclareRobustCommand{\eg}{\textit{e.g.}}
\DeclareRobustCommand{\ie}{\textit{i.e.}}
\DeclareRobustCommand{\cf}{\textit{cf.}}
\DeclareRobustCommand{\wrt}{\textit{w.r.t.}}
\makeatletter
\DeclareRobustCommand{\etc}{%
    \@ifnextchar{.}%
        {\textit{etc}}%
        {\textit{etc.}}%
}
\makeatother

\newcommand{\mbb}[1]{{\ensuremath{\mathbb{#1}}}}
\newcommand{\mcl}[1]{{\ensuremath{\mathcal{#1}}}}
\newcommand{\mfk}[1]{{\ensuremath{\mathfrak{#1}}}}

\newcommand{\mrm}[1]{{\ensuremath{\mathrm{#1}}}}

\newcommand{\ltrue}{\tt t\!t}
\newcommand{\lfalse}{\tt f\!f}

\newcommand{\dif}{\mathrm{d}}

\newcommand{\cone}[1]{{#1{\uparrow}}}
\newcommand{\scone}[1]{{#1{\downarrow}}}

\newcommand{\fw}{\mathcal{F}_{\mathfrak{W}}}

\makeatletter
\newcommand{\xRightarrow}[2][]{\ext@arrow 0359\Rightarrowfill@{#1}{#2}}
\makeatother

\newcommand{\defeq}{\triangleq}
\newcommand{\defiff}{\stackrel{\triangle}{\iff}}

\newcommand{\cat}[1]{\textsc{#1} }

\newcommand{\hole}{\mbox{\large\bf\_} }

\hypersetup{
	pdftitle={Weak bisimulations for labelled transition systems weighted over semirings},
	pdfauthor={Marino Miculan, Marco  Peressotti},
	pdfsubject={Concurrency Theory},
	pdfkeywords={Weighted systems, weak bisimulation}
}

\title{Weak bisimulations for labelled transition systems\\ weighted over semirings}
\author{Marino Miculan \qquad Marco  Peressotti\\
Dept.~of Mathematics and
  Computer Science, University of Udine, Italy
  \\
  \href{mailto:marino.miculan@uniud.it}{\tt marino.miculan@uniud.it}
  \qquad
   \href{mailto:marco.peressotti@uniud.it}{\tt marco.peressotti@uniud.it}
}
\date{October 10, 2013}

\begin{document}

\maketitle

\begin{abstract}
  Weighted labelled transition systems are LTSs whose transitions are
  given weights drawn from a commutative monoid.  WLTSs subsume a wide
  range of LTSs, providing a general notion of strong (weighted)
  bisimulation.  In this paper we extend this framework towards other
  behavioural equivalences, by considering \emph{semirings} of
  weights.  Taking advantage of this extra structure, we introduce a
  general notion of \emph{weak weighted bisimulation}. We show that
  weak weighted bisimulation coincides with the usual weak
  bisimulations in the cases of non-deterministic and
  fully-probabilistic systems; moreover, it naturally provides a
  definition of weak bisimulation also for kinds of LTSs where this
  notion is currently missing (such as, stochastic systems).  Finally,
  we provide a categorical account of the coalgebraic construction of
  weak weighted bisimulation; this construction points out how to port
  our approach to other equivalences based on different notion of
  observability.
\end{abstract}

\section{Introduction}
Many extensions of labelled transition systems have been proposed for
dealing with quantitative notions such as execution times, transition
probabilities and stochastic rates; see \eg~\cite{bg98:empa,denicola13:ultras,hhk2002:tcs,hillston:pepabook,klinS08,pc95:cj}
among others.  This ever-increasing plethora of variants has naturally
pointed out the need for general mathematical frameworks, covering
uniformly a wide range of cases, and offering general results and
tools.  As examples of these theories we mention \textsc{ULTraS}s
\cite{denicola13:ultras} and \emph{weighted labelled transition
  systems} (WLTSs) \cite{tofts1990:synchronous,Klin09wts,klins:ic2012}. In particular, in a
WLTS every transition is associated with a \emph{weight} drawn from a
commutative monoid \mfk W; the monoid structure defines how weights of
alternative transitions combine.  As we will recall in
Section~\ref{sec:wlts}, by suitably choosing this monoid we can
recover ordinary
non-deterministic LTSs, probabilistic transition systems, and
stochastic transition systems, among others. WLTSs offer
a notion of \emph{(strong) \mfk W-weighted bisimulation}, which can
be readily instantiated to particular cases obtaining precisely the
well-known Milner's strong bisimulation \cite{milner:cc}, Larsen-Skou's
strong probabilistic bisimulation \cite{ls:probbisim}, strong stochastic
bisimulation \cite{hillston:pepabook}, \etc.

However, in many situations strong bisimulations are too fine, and
many coarser relations have been introduced since then.  Basically,
these \emph{observational} equivalences do not distinguish systems
differing only for unobservable or not relevant transitions.
Likely the most widely known of these observational equivalences is
Milner's \emph{weak bisimulation} for non-deterministic LTSs
\cite{milner:cc} (but see
\cite{glabbeek90:spectrum,glabbeek93:spectrum2} for many variations).
Weak bisimulations focus on systems' interactions (communications,
synchronizations, etc.), ignoring transitions associated with
systems' internal operations, hence called \emph{silent} (and denoted
by the $\tau$). 

\looseness=-1
Unfortunately, weak bisimulations become quite more problematic in
models for stochastic systems, probabilistic systems, \etc. The
conundrum is that we do not want to observe $\tau$-transitions but at
the same time their quantitative effects (delays, probability
distributions) are still observable and hence cannot be ignored.  In
fact, for quantitative systems there is no general agreement of what a
weak bisimulation should be.  As an example, consider the stochastic
system $S_1$ executing an action $a$ at rate $r$, and a system $S_2$
executing $\tau$ at rate $r_1$, followed by an $a$ at rate $r_2$:
should these two systems be considered weakly bisimilar?  
\[\begin{tikzpicture}[auto,xscale=2,font=\footnotesize,scale=.9,
	dot/.style={circle,	minimum size=5pt,inner sep=0pt, outer sep=2pt},
	dot stage1/.style={dot,fill=white,draw=black},
	dot stage2/.style={dot,fill=black,draw=black},
	arr stage1/.style={->,-open triangle 60},	
	arr stage2/.style={->,-triangle 60},
	baseline=(current bounding box.center)]
	\begin{scope}
	\node[dot stage1,label=180:$S_1$] (n0) at (0,0) {};
	\node[dot stage1] (n1) at (1,0) {};

	\draw[arr stage1] (n0) to node {\(a,r\)} (n1);
	\end{scope}

	\begin{scope}[xshift=25mm]
	\node[dot stage1,label=180:$S_2$] (n0) at (0,0) {};
	\node[dot stage1] (n1) at (1,0) {};
	\node[dot stage1] (n2) at (2,0) {};

	\draw[arr stage1] (n0) to node {\(\tau,r_1\)} (n1);
	\draw[arr stage1] (n1) to node {\(a,r_2\)} (n2);
	\end{scope}
\end{tikzpicture}\]
Some
approaches restrict to instantaneous $\tau$-actions (and hence
$r_2=r$) \cite{mb2007:ictcs}; others require that the average times of
$a$'s executions are the same the two systems - but still these can be
distinguished by looking at the variances \cite{mb2012:qapl}.
Therefore, it is not surprising that many definitions proposed in
literature are rather \emph{ad-hoc}, and that a general mathematical
theory is still missing.

This is the problem we aim to address in this paper. More precisely,
in Section~\ref{sec:weak-bisim} we introduce the uniform notion of
\emph{weak weighted bisimulation} which applies to labelled transition
systems weighted over a \emph{semiring}.  The multiplication operation of
semirings allows us to compositionally extend weights to multi-step
transitions and traces.  In Section~\ref{sec:weak-instances} we show
that our notion of weak bisimulation coincides with the known ones in
the cases of non-deterministic and fully probabilistic systems, just
by changing the underlying semiring.  Moreover it naturally applies to
stochastic systems, providing an effective notion of \emph{weak
  stochastic bisimulation}.  As a side result we introduce a new
semiring of \emph{stochastic variables} which generalizes that of
rated transition systems \cite{klinS08}.

Then, in Section~\ref{sec:cw-algorithm} we present the general
algorithm for computing weak weighted bisimulation equivalence
classes, parametric in the underlying semiring. This algorithm is a
variant of Kanellakis-Smolka's algorithm for deciding strong
non-deterministic bisimulation \cite{ks1990:ic}.  Our solution builds
on the refinement technique used for the \emph{coarsest stable
  partition}, but instead of ``strong'' transitions in the original
system we consider ``weakened'' ones.  We prove that this algorithm is
correct, provided the semiring satisfies some mild conditions,
\ie\ it is $\omega$-complete. Finally, we discuss also its complexity,
which is comparable with Kanellakis-Smolka's algorithm. Thus, this
algorithm can be used in the verification of many kinds of systems,
just by replacing the underlying semiring (boolean, probabilistic,
stochastic, tropical, arctic, \dots) and taking advantage of existing
software packages for linear algebras over semirings.

In Section~\ref{sec:cat-view} we give a brief categorical account of
weak weighted bisimulations. These will be characterized as
cocongruences between suitably \emph{saturated} systems, akin to the
elegant construction of $\epsilon$-elimination given in
\cite{sw2013:epsilon}.

In Section~\ref{sec:concl} we give some final remarks and directions for further work.

\section{Weighted labelled transition systems}\label{sec:wlts}

In this section we recall the notion of \emph{labelled transition
  systems weighted over a commutative monoid}, showing how these
subsume non-deterministic, stochastic and probabilistic systems, among
many others.
Weighted LTSs were originally introduced by Klin in \cite{Klin09wts} as
the prosecution and generalization of the work on stochastic SOS
presented in \cite{klinS08} with Sassone and were further developed in
\cite{klins:ic2012}.

In the following let \mfk W denote a generic commutative (aka \emph{abelian}) monoid
$(W,+,0)$, \ie~a set equipped with a distinguished element $0$ and a 
binary operation $+$ which is associative, commutative and has $0$ as 
left and right unit.

\begin{definition}[\mfk W-LTS {\cite[Def.~1]{Klin09wts}}]
\label{def:wlts}
  Given a commutative monoid $\mfk W = (W,+,0)$, a \emph{\mfk W-weighted labelled 
  transistion system} is a triple $(X,A,\rho)$ where:
  \begin{itemize}
     \item $X$ is a set of \emph{states} (processes);
     \item $A$ is an at most countable set of \emph{labels};
     \item $\rho:X\times A \times X \to W$ is a \emph{weight
         function}, mapping each triple of $X\times A \times X$ to a
       weight.
  \end{itemize}
  $(X,A,\rho)$ is said to be \emph{image finite}
  (resp.~\emph{countable}) iff for each $x\in X$ and $a\in A$, the set
  $\{ y \in X \mid \rho(x,a,y)\neq 0\}$ is finite (resp.~countable).
  A state $x\in X$ is said \emph{terminal} iff for every $a \in A$ and
  $y\in X$: $\rho(x,a,y)=0$.
\end{definition}
For adherence to the notation used in \cite{Klin09wts} and to support
the intuitions based on classical labelled transition systems we shall
often write $\rho(x\xrightarrow{a}y)$ for $\rho(x,a,y)$; moreover,
following a common notation for stochastic and probabilistic systems,
we will write also $x\xrightarrow{a,w}y$ to denote $\rho(x,a,y) = w$.

The monoidal structure was not used in Definition~\ref{def:wlts} but
for the existence of a distinguished element required by the image
finiteness (resp. countability) property.  The commutative monoidal
structure of weights comes into play in the notion of
\emph{bisimulation}, where weights of transitions with the same labels
have to be ``summed''. This operation is commonplace for stochastic
LTSs, but at first it may appear confusing with respect to the notion of
bisimulation of non-deterministic LTSs; we will explain it in
Section~\ref{sec:lts-as-wlts}.

\begin{definition}[Strong \mfk W-bisimulation {\cite[Def.~3]{Klin09wts}}]
\label{def:wlts-strong}
  Given a \mfk W-LTS $(X,A,\rho)$, a \emph{(strong) \mfk W-bisimulation} is an
  equivalence relation $R$ on $X$ such that for each pair $(x,x')$ of
  elements of $X$, $(x,x') \in R$ implies that for each label $a\in A$
  and each equivalence class $C$ of $R$:
  \[\sum_{y\in C}\rho(x\xrightarrow{a}y) = \sum_{y\in C}\rho(x'\xrightarrow{a}y)\text.\]
  Processes $x$ and $x'$ are said to be \emph{\mfk W-bisimilar} 
  (or just bisimilar when \mfk W is understood) if there exists a 
  \mfk W-bisimulation $\sim_{\mfk W}$ such that $x\sim_{\mfk W} x'$.
\end{definition}
Clearly \mfk W-bisimulations are closed under arbitrary unions
ensuring the \emph{\mfk W-bisimilarity} on any \mfk W-LTS to be the
largest \mfk W-bisimulation over it.\footnote{Actually, strong \mfk
  W-bisimulation has been proven to be a strong bisimulation in
  coalgebraic sense \cite{Klin09wts}.}

\begin{remark}%
\label{rem:wlts-strong}
In order for the above definition to be well-given, summations
need to be well-defined. Intuitively this means that
the \mfk W-LTS $(X,A,\rho)$ does not exceed the expressiveness 
of its underlying monoid of weights \mfk W.  Reworded, the
system has to be image finite if the monoid admits only finite
summations; image countable if the monoid admits countable
summations, and so on.
\end{remark}

In \cite{Klin09wts,klins:ic2012}, for the sake of simplicity the
authors restrict themselves to image finite systems (which is not
unusual in the coalgebraic setting).  In the present paper we extend
their definitions to the case of countable images.  This
generalization allows to capture a wider range of systems and is
crucial for the definition of weak and delay bisimulations.

In practice, Remark \ref{rem:wlts-strong} is not a severe restriction,
since the commutative monoids relevant for most systems of interest
admit summations over countable sets.  To supports this claim, in the
rest of this Section we illustrate how non-deterministic, stochastic
and probabilistic labelled transition systems can be recovered as
systems weighted over commutative monoids whit countable sums.  These
kind of commutative monoids are often called \emph{commutative
  $\omega$-monoids}\footnote{Monoids can be readily extended to
  $\omega$-monoids adding either colimits freely or an ``$\infty$''
  element.}.

\subsection{Non-deterministic systems are WLTS}
\label{sec:lts-as-wlts}
This section illustrates how non-deterministic labelled transition systems 
\cite{milner:cc}
can be recovered as systems weighted over the commutative $\omega$-monoid of 
logical values equipped with logical disjunction 
$\mfk 2 \defeq (\{\ltrue,\lfalse\},\lor,\lfalse)$.

\begin{definition}[Non-deterministic LTS]
\label{def:lts}
  A \emph{non-deterministic labelled transition system}
  is a triple $(X,A,\rightarrow)$ where:
  \begin{itemize}
     \item $X$ is a set of \emph{states} (processes);
     \item $A$ is an at most countable set of \emph{labels} (actions);
     \item $\mathop{\rightarrow} \subseteq X\times A\times X$ is the \emph{transition
       relation}.
  \end{itemize}
  As usual, we shall denote an $a$-labelled transition from $x$ to $y$
  \ie~$(x,a,y) \in {\rightarrow}$ by $x\xrightarrow{a}y$.  A state $y$
  is called \emph{successor} of a given state $x$ iff
  $x\xrightarrow{a}y$. If $x$ has no successors then it is said to be
  \emph{terminal}.  If every state has a finite set of
  successors then the system is said to be \emph{image finite}.
  Likewise it is said to be \emph{image countable} if each state has
  at most countably many successors.
\end{definition}

Every \mfk 2-valued weight function is a predicate defining a subset
of its domain, turning $\rho: X\times A \times X \to \mfk 2$
equivalent to the classical definition of the transition relation
${\rightarrow} \subseteq X\times A \times X$.

\begin{definition}[Strong non-deterministic bisimulation]
\label{def:lts-bisim}
  Let $(X,A,{\rightarrow})$ be an LTS. An equivalence relation $R\subseteq X\times X$
  is a \emph{(strong non-deterministic) bisimulation on $(X,A,\rightarrow)$} iff 
  for each pair of states $(x,x') \in R$, for any label 
  $a \in A$ and each equivalence class $C \in X/R$:
  \[
    \exists y \in C. x\xrightarrow{a}y \iff \exists y' \in C. x'\xrightarrow{a}y'
  \text.\]
  Two states $x$ and $x'$ are said \emph{bisimilar} iff there exists a bisimulation
  relation $\sim_l$ such that $x\sim_l y$.  The greatest bisimulation for 
  $(X,A,\rightarrow)$ uniquely exists and is called \emph{(strong) bisimilarity}.
\end{definition}
Strong \mfk 2-bisimulation and strong non-deterministic bisimulation
coincide, since logical disjunction over the states in a given class
$C$ encodes the ability to reach $C$ making an $a$-labelled transition.

\subsection{Stochastic systems are WLTS}
\label{sec:rts-wlts}
Stochastic systems have important application especially in the field
of quantitative analysis, and several tools and formalisms to describe
and study them have been proposed (\eg~PEPA
\cite{hillston:pepabook}, EMPA \cite{bg98:empa} and the stochastic
$\pi$-calculus \cite{pc95:cj}).  Recently, \emph{rated transition
  systems} \cite{klinS08,klins:ic2012,denicola09:rts,denicola13:ultras} emerged
as a convenient presentation of these kind of systems.

\begin{definition}[Rated LTS {\cite[Sec.~2.2]{klinS08}}]
\label{def:rlts}
  A \emph{rated labelled transition system}
  is a triple $(X,A,\rho)$ where:
  \begin{itemize}
     \item $X$ is a set of \emph{states} (processes);
     \item $A$ is a countable set of \emph{labels} (actions);
     \item $\rho : X\times A\times X \to \mbb R^+_0$ is the \emph{rate function}.
  \end{itemize}
\end{definition}

Semantics of stochastic processes is usually given by means of labelled
continuous time Markov chains (CTMC).  The real number $\rho(x,a,y)$
is interpreted as the parameter of an exponential probability
distribution governing the duration of the transition from state $x$
to $y$ by means of an $a$-labelled action and hence encodes the
underlying CTMC (for more information about CTMCs and their presentation
by transition rates see \eg~\cite{hhk2002:tcs,hillston:pepabook,prakashbook09,pc95:cj}).

\begin{definition}[Strong stochastic bisimulation]
\label{def:rts-bisim}
  Given a rated system $(X,A,\rho)$ an equivalence relation $R\subseteq X\times X$
  is a \emph{(strong stochastic) bisimulation on $(X,A,\rho)$} 
  (or \emph{strong equivalence} \cite{hillston:pepabook}) iff 
  for each pair of states $(x,x') \in R$, for any label 
  $a \in A$ and each equivalence class $C \in X/R$:
  \[\sum_{y\in C}\rho(x,a,y) = \sum_{y\in C}\rho(x',a,y)\text.\]
  Two states $x$ and $x'$ are said \emph{bisimilar} iff there exists a bisimulation
  relation $\sim_s$ such that $x\sim_s y$. The greatest bisimulation for 
  $(X,A,\rightarrow)$ uniquely exists and is called \emph{(strong) bisimilarity}.
\end{definition}

Rated transition systems (hence stochastic systems) are precisely
WLTS weighted over the commutative monoid of nonnegative real 
numbers (closed with infinity) under addition 
$(\overline{\mbb R}_0^+,+,0)$ and stochastic bisimulations
correspond to $\overline{\mbb R}_0^+$-bisimulations, as shown in \cite{Klin09wts}.
Moreover, $\overline{\mbb R}_0^+$ is an $\omega$-monoid since non-negative real numbers
admit sums over countable families.
In particular, the sum of a given countable family $\{x_i \mid i \in I\}$
is defined as the supremum of the set of sums over its finite subfamilies:
\[\sum_{i\in I}x_i \defeq \sup\left\{\sum_{i \in J}x_i \mid J \subseteq I, |J| < \omega \right\}\text.\]

\subsection{Probabilistic systems are (Constrained) WLTS}
\label{ex:plts-as-wlts}

This section illustrates how probabilistic LTSs are captured by
weighted ones.  We focus on fully probabilistic systems (also known as
\emph{generative systems}) \cite{gsb90:ic,ls:probbisim, baier97:cav}
but in the end we provide some hints on other types of probabilistic
systems.

Fully probabilistic system can be regarded as a specializations of
non-deterministic transition systems where probabilities are used to
resolve nondeterminism.  From a slightly different point of view, they
can also be interpreted as labelled Markov chains with discrete
parameter set \cite{ks76:fmc}.

\begin{definition}[Fully probabilistic LTS]
\label{def:plts}
  A \emph{fully probabilistic labelled transition system}
  is a triple $(X,A,\mrm P)$ where:
  \begin{enumerate}[\em(1)]
     \item $X$ is a set of \emph{states} (processes);
     \item $A$ is a countable set of \emph{labels} (actions);
     \item $\mrm P : X\times A\times X \to [0,1]$ is a function such
       that for any $x \in X$ $\mrm P(x,\hole, \hole)$ is either a
       discrete probability measures for $A \times X$ or the
       constantly $0$ function.
  \end{enumerate}
\end{definition}
In ``reactive'' probabilistic systems, in contrast to fully probabilistic systems, 
transition probability distributions are dependent on the occurrences of actions
\ie~for any $x \in X$ and $a \in A$ $\mrm P(x,a, \hole)$ is either a discrete 
probability measures for $X$ or the constantly $0$ function.

Strong probabilistic bisimulation has been originally introduced by Larsen and 
Skou \cite{ls:probbisim} for reactive systems and has been reformulated by 
van Glabbeek et al.~\cite{gsb90:ic} for fully probabilistic systems.

\begin{definition}[Strong probabilistic bisimilarity]
\label{def:plts-bisim}
  Let $(X,A,\mrm P)$ be a fully probabilistic system.  An equivalence relation 
  $R\subseteq X\times X$ is a \emph{(strong probabilistic) bisimulation on 
  $(X,A,\mrm P)$} iff for each pair of states $(x,x') \in R$, 
   for any label $a \in A$ and any equivalence class $C \in X/R$:
  \[
    \mrm P(x,a,C) = \mrm P(x',a,C)
  \]
  where $\mrm P(x,a,C) \defeq \sum_{y \in C}\mrm P(x,a,y)$.

  Two states $x$ and $x'$ are said \emph{bisimilar} iff there exists a bisimulation
  relation $\sim_p$ such that $x\sim_p y$. The greatest bisimulation for 
  $(X,A,\mrm P)$ uniquely exists and is called \emph{bisimilarity}.
\end{definition}

It would be tempting to recover fully probabilistic systems as LTS
weighted over the probabilities interval $[0,1]$ but unfortunately the
addition does not define a monoid on $[0,1]$ since it is not a total
operation when restricted $[0,1]$.  There exist various commutative
monoids over the probabilities interval, leading to different
interpretations of probabilistic systems (as will be shown in Section
\ref{sec:other-semirings}), but since in
Definition~\ref{def:plts-bisim} we sum probabilities of outgoing
transitions (\eg~to compute the probability of reaching a certain set
of states), the real number addition has to be used.

\begin{remark}[On partial commutative monoids]\label{rem-partial}
  The theory of weighted labelled transition systems can be extended to
  consider partial commutative monoids (\ie~$a + b$ may be undefined
  but when it is defined then also $b + a$ is and commutativity
  holds) or commutative $\sigma$-monoids to handle sums over
  opportune countable families (thus relaxing the requirement of
  weights forming $\omega$-monoids). However, every
  $\sigma$-semiring can be turned into an $\omega$-complete one
  by adding a distinguished $+\infty$ element and resolving
  partiality accordingly.
\end{remark}

Klin \cite{Klin09wts} suggested to consider probabilistic systems as
systems weighted over $(\mathbb R_0^+,+,0)$ but subject to suitable
constraints ensuring that the weight function is a state-indexed
probability distribution and thus satisfies Definition~\ref{def:plts}.
These \emph{constrained} WLTSs were proposed %
to deal with reactive probabilistic systems.
\begin{definition}[constrained \mfk W-LTS]
\label{def:cwlts}
  Let \mfk W be a commutative monoid and \mcl C be a constraint family.
  A \emph{$\mcl C$-constrained \mfk W-weighted labelled transistion system} is a 
  \mfk W-LTS $(X,A,\rho)$ such that its weight function $\rho$
  satisfies the constraints \mcl C over \mfk W.
\end{definition}

Then, fully probabilistic labelled transition systems 
are precisely constrained $\mathbb R_0^+$-LTSs $(X,A,\rho)$ subject to the 
constraint family:
\[\sum_{a \in A, y \in X}\rho(x,a,y) \in \{0,1\}\mbox{  for }x\in X\text.\]
Likewise, reactive probabilistic systems are $\mathbb R_0^+$-LTSs subject to the 
constraint family:
\[\sum_{y \in X}\rho(x,a,y) \in \{0,1\}\mbox{  for $x\in X$ and $a\in A$}\text.\]
Therefore strong bisimulations for these kind of systems are exactly strong 
$\mathbb R_0^+$-bisimulations.

\section{Weak bisimulations for WLTS over semirings}\label{sec:weak-bisim}

In the previous section we illustrated how weighted labelled
transitions systems can uniformly express several kinds of systems
such as non-deterministic, stochastic and probabilistic systems.
Remarkably, bisimulations for these systems were proved to be
instances of weighted bisimulations.

In this section we show how other observational equivalences can
be stated at the general level of the weighted transition system
offering a treatment for these notions uniform across the wide range
of systems captured by weighted ones.  Due to space constraints we
focus on weak bisimulation but eventually we discuss briefly how the proposed
results can cover other notions of observational equivalence.

\subsection{From transitions to execution paths}

Let $(X,A+\{\tau\},\rho)$ be a \mfk W-LTS. 
A \emph{finite execution path} $\pi$ for this system is a sequence of transition \ie~an 
alternating sequence of states and labels like
\[\pi = x_0 \xrightarrow{a_1} x_1 \xrightarrow{a_2} x_2 \dots x_{n-1}\xrightarrow{a_n} x_n\]
such that for each transition $x_{i-1} \xrightarrow{a_i} x_i$ in the path:
\[
  \rho(x_{i-1} \xrightarrow{a_i} x_i)\neq 0.
\]
Let $\pi$ denote the above path, then set:
\begin{gather*}
\mrm{length}(\pi) = n\qquad \mrm{first}(\pi) = x_0\qquad \mrm{last}(\pi) = x_n \qquad
\mrm{trace}(\pi) = a_1a_2\dots a_n
\text.\end{gather*}
to denote the length, starting state, ending state and trace of $\pi$ respectively.

In order to extend the definition of the weight function $\rho$ to
executions we need some additional structure on the domain of weights,
allowing us to capture concatenation of transition.  To this end, we
require weights to be drawn from a semiring, akin to the theory of
  weighted automata.  Recall that a semiring is a set $W$
equipped with two binary operations $+$ and $\cdot$ called
\emph{addition} and \emph{multiplication} respectively and such that:
\begin{itemize}
\item $(W,+,0)$ is a commutative monoid and $(W,\cdot,1)$ is a monoid;
\item multiplication left and right distributes over addition:
\begin{gather*}
a\cdot(b+c) = (a\cdot b) + (a\cdot c)\qquad
(a+b)\cdot c = (a\cdot c) + (b\cdot c)
\end{gather*}
\item multiplication by $0$ annihilates $W$:
\[0\cdot a = 0 = a \cdot 0\text.\]
\end{itemize}

Basically, the idea is to express parallel and subsequent transitions
(\ie~branching and composition) by means of addition and
multiplication respectively.  Therefore, multiplication is not
required to be commutative (\cf~the semiring of formal languages).
Distributivity ensures that execution paths are independent from the
alternative branching \ie~given two executions sharing some sub-path,
we are not interested in which is the origin of the sharing; as the
following diagram illustrates:
\begin{equation}\label{eq:path-dist}
\begin{tikzpicture}[auto,font=\small,
  baseline=(current bounding box.center),
  extended/.style={shorten >=-#1, shorten <=-#1},
  extended/.default=0pt,
  dot/.style={circle,fill=black,minimum size=4pt,inner sep=0, outer sep=2pt}]
  \node (s0) at (0,0) {
  \begin{tikzpicture}[auto,yscale=0.6,xscale=.5]
    \node[dot] (n0) at (0,2) {\(\)};
    \node[dot] (n1) at (0,1) {\(\)};
    \node[dot] (n2) at (-.7,0) {\(\)};
    \node[dot] (n3) at (0.7,0) {\(\)};
    \draw[->]  (n0) to node [] {\(a\)} (n1);
    \draw[->]  (n1) to node [pos=.7,swap] {\(b\)} (n2);
    \draw[->]  (n1) to node [pos=.7] {\(c\)} (n3);
  \end{tikzpicture}
  };
  \node[right=16pt of s0] (s1) {
  \begin{tikzpicture}[auto,yscale=0.6,xscale=.5]
    \node[dot] (n0) at (0,2) {\(\)};
    \node[dot] (n1) at (-.7,1) {\(\)};
    \node[dot] (n2) at (-.7,0) {\(\)};
    \node[dot] (n3) at (0,2) {\(\)};
    \node[dot] (n4) at (.7,1) {\(\)};
    \node[dot] (n5) at (.7,0) {\(\)};
    \draw[->]  (n0) to node [pos=.7,swap] {\(a\)} (n1);
    \draw[->]  (n1) to node [swap] {\(b\)} (n2);
    \draw[->]  (n3) to node [pos=.7] {\(a\)} (n4);
    \draw[->]  (n4) to node [] {\(c\)} (n5);
  \end{tikzpicture}
  };
  \node[right=23pt of s1] (s2) {
  \begin{tikzpicture}[auto,yscale=0.6,xscale=.5]
    \node[dot] (n0) at (0,0) {\(\)};
    \node[dot] (n1) at (0,1) {\(\)};
    \node[dot] (n2) at (-.7,2) {\(\)};
    \node[dot] (n3) at (0.7,2) {\(\)};
    \draw[<-]  (n0) to node [swap] {\(c\)} (n1);
    \draw[<-]  (n1) to node [pos=.7] {\(a\)} (n2);
    \draw[<-]  (n1) to node [pos=.7,swap] {\(b\)} (n3);
  \end{tikzpicture}
  };
  \node[right=16pt of s2] (s3) {
  \begin{tikzpicture}[auto,yscale=0.6,xscale=.5]
    \node[dot] (n0) at (-.7,2) {\(\)};
    \node[dot] (n1) at (-.7,1) {\(\)};
    \node[dot] (n2) at (0,0) {\(\)};
    \node[dot] (n3) at (.7,2) {\(\)};
    \node[dot] (n4) at (.7,1) {\(\)};
    \node[dot] (n5) at (0,0) {\(\)};
    \draw[->]  (n0) to node [swap] {\(a\)} (n1);
    \draw[->]  (n1) to node [pos=.7,swap] {\(c\)} (n2);
    \draw[->]  (n3) to node [] {\(b\)} (n4);
    \draw[->]  (n4) to node [pos=.7] {\(c\)} (n5);
  \end{tikzpicture}
  };
  \node at ($(s0)!0.5!(s1)$) {\(\equiv\)};
  \node at ($(s2)!0.5!(s3)$) {\(\equiv\)};
\end{tikzpicture}
\end{equation}
Finally, since weights of (proper) transitions are always different from $0$,
the annihilation property means that no proper execution can contain 
improper transitions.

Then, the weight function $\rho$ extends to finite paths by semiring multiplication
(therefore we shall use the same symbol):
\[\rho(x_0 \xrightarrow{a_1}x_1\dots \xrightarrow{a_n} x_n) \defeq \prod_{i=1}^n\rho(x_{i-1},a_i,x_i)\]

In the following let \mfk W be a semiring $(W,+,0,\cdot,1)$.

Semirings offer enough structure to extend weight function to finite execution
paths compositionally but executions can also be (countably) infinite.
Likewise countable branchings (\cf~Remark \ref{rem:wlts-strong}), paths of countable
length can be treated requiring multiplication to be defined also over 
(suitable) countable families of weights and obviously respect the
semiring structure.
However, the additional requirement for $(W,\cdot,1)$ can be avoided
by dealing with suitable sets of paths as long as these convey enough
information for the notion of weak bisimulation (and observational
equivalence in general).  In particular, a finite path $\pi$
determines a set of paths (possibly infinite) starting with $\pi$,
thus $\pi$ can be seen as a representative for the set.  Moreover, the
behavior of a system can be reduced to its \emph{complete} executions:
a path is called \emph{complete} (or ``full'' \cite{baier97:cav}) if
it is either infinite or ends in a \emph{terminal} state.

Intuitively, we distinguish complete paths only up to the chosen
representatives: longer representative may generate smaller sets of
paths, and this can be thought in ``observing more'' the system. If
two complete paths are distinguishable, we have to be able to
distinguish them in a finite way \ie~there must be two representative
with enough information to tell one set from the other.  Otherwise, if
no such representative exist, then the given complete paths are indeed
equivalent.  Therefore, it is enough to be able to compositionally
weight (finite) representatives in order to distinguish any complete
path.

The remaining of the subsection elaborates the above intuition
defining a $\sigma$-algebra over complete paths (for each state).
The method presented is a generalization to semirings of the one used in
\cite{baier98}.  This structure allows to deal
with sets of finite paths avoiding redundancies 
(\cf~Example \ref{ex:class-reachability})
and define weights compositionally.

Let $\mrm{Paths}(x)$, $\mrm{CPaths}(x)$ and $\mrm{FPaths}(x)$ denote the sets of
all, complete and finite paths starting in the state $x \in X$ respectively.
Likewise, we shall denote the corresponding sets of paths \wrt~any starting state
as $\mrm{Paths}$, $\mrm{CPaths}$ and $\mrm{FPaths}$ respectively (\eg~$\mrm{Paths} = \cup_{x\in X} \mrm{Paths}(x)$).
Paths naturally organize into a preorder by the prefix relation. In particular,
given $\pi,\pi' \in \mrm{Paths}(x)$ define $\pi\preceq\pi'$ if and only if 
one of the following holds:
\begin{enumerate}
  \item $\pi \equiv 
    x \xrightarrow{a_1} x_1 \dots \xrightarrow{a_n} x_n$ and
    $\pi' \equiv 
    x \xrightarrow{a'_1} x'_1  \dots \xrightarrow{a'_n} x'_{n'}$
    (both finite),
    $x_i = x'_i$ and $a_i = a'_i$ for $i \leq n \leq n'$;
  \item $\pi \equiv 
    x \xrightarrow{a_1} x_1 \dots \xrightarrow{a_n} x_n$ and
    $\pi' \equiv x \xrightarrow{a'_1} x'_1  \dots$
    (one finite and the other infinite),
    $x_i = x'_i$ and $a_i = a'_i$ for $i \leq n$;
  \item $\pi = \pi'$ (both infinite).
\end{enumerate}
For each finite path $\pi \in \mrm{FPaths}(x)$ define the \emph{cone
of complete paths generated by $\pi$} as follows:
\[
  \cone\pi \defeq \{\pi'\in \mrm{CPaths}(x) \mid \pi \preceq \pi'\}
\text.\]
Cones are precisely the sets we were sketching in the intuition above
and form a subset of the parts of $\mrm{CPaths}(x)$:
\[
\Gamma \defeq \{\cone\pi \mid \pi \in \mrm{FPaths}(x)\}\text.
\]
This set is at most countable since the set $\mrm{FPaths}(x)$ is so
and every two of its elements are either disjoint or one the subset 
of the other as the following Lemmas state.

\begin{lemma}
\label{lem:fpath-countable}
For any state $x\in X$, the set of finite paths $\mrm{FPaths}(x)$
of an image countable \mfk W-LTS is at most countable.
\end{lemma}
\begin{proof} 
By induction on the length $k$ of paths in $\mrm{FPaths}(x)$, these
are at most countable. In fact, for $k = 0$ there is exactly one path,
$\varepsilon$ and, taken the set of paths of length $k$ be at most 
countable, then the set of those with length $k+1$ is at most countable
because the system is assumed to be image countable. 
Then $\mrm{FPaths}(x)$ is at most countable since it is the disjoint union of 
\[\{\pi \in \mrm{Paths}(x) \mid \mrm{length}(\pi) = k\}\]
for $k \in \mbb N$.
\end{proof}

\begin{lemma}
\label{lem:geom-cones}
Two cones $\cone\pi_1$ and $\cone{\pi_2}$ are either disjoint
or one the subset of the other.
\end{lemma}
\begin{proof} 
For any $\pi \in \mrm{CPaths}(x)$, we have by definition:
\[\pi\in\cone{\pi_1} \iff \pi_1\preceq \pi \mbox{ and } \pi\in\cone{\pi_2} \iff 
\pi_2\preceq\]
Then, if $\pi_1\preceq\pi_2$ then $\pi\in\cone{\pi_2} \Rightarrow \pi\in\cone{\pi_1}$
(likewise for $\pi_2\preceq\pi_1$). For the other case, since
$\pi_1\not\preceq\pi_2 \land \pi_2\not\preceq\pi_1$ there is no $\pi$ such that
$\pi_1\preceq \pi \land \pi_2\preceq \pi$.
\end{proof}

Given $\Pi\subseteq\mrm{FPaths}(x)$, the set of all cones generated by
its elements is denoted by $\cone\Pi$ and defined as the (at most
countable) union of the cones generated by each $\pi \in \Pi$.  If
this union is over disjoint cones then $\Pi$ is said to be
\emph{minimal}.

Minimality is not preserved by set union even if operands are disjoint and
both minimal. As a counter example consider the sets $\{\pi\}$ and $\{\pi'\}$
for $\pi \prec \pi' \in \mrm{FPaths}(x)$; both are minimal and disjoint, but
their union is not minimal since $\cone{\pi'}\subseteq\cone{\pi}$.
However, $\Pi$ always has at least a subset $\Pi'$ being minimal and such that
\begin{equation}\label{eq:gen-same-cones}
  \cone\Pi = \cone{\Pi'}
\text.\end{equation}
and among these there exists exactly one which is also minimal in the sense of prefixes:
\begin{lemma}\label{lem:fpaths-support}
For $\Pi \subseteq \mrm{FPaths}(x)$, there exists a minimal subset 
$\Pi'\subseteq\Pi$ which satisfies (\ref{eq:gen-same-cones}), \ie~for any
$\Pi''\subseteq\Pi$ satisfying (\ref{eq:gen-same-cones}) we have:
$
\forall \pi''\in\Pi''\, \exists \pi' \in \Pi . \pi'\preceq\pi''\text.
$
We denote such $\Pi'$ by $\scone\Pi$.
\end{lemma}
\begin{proof}
Clearly $\cone\Pi = \emptyset$ iff $\Pi = \emptyset$ since there are no infinite
prefix descending chains. 
Then $\cone{(\scone{\Pi})} \subseteq \cone\Pi$ since $\scone{\Pi} \subseteq \Pi$ is minimal. For every $\pi \in \Pi$ there exists
$\pi' \in \scone\Pi$ such that $\pi'\preceq\pi$ and by 
Lemma~\ref{lem:geom-cones} $\cone\pi\subseteq\cone{\pi'}$ \ie~
$\cone\Pi\subseteq\cone{(\scone{\Pi})}$. Therefore $\cone\Pi = \cone{(\scone{\Pi})}$.
Consider $\Pi'$ as in the enunciate, then, for every $\pi\in \Pi$ there exists
$\pi' \in \scone\Pi$ such that $\pi'\preceq\pi$ and in particular if $\pi \in \Pi'$.
Uniqueness follows straightforwardly.
\end{proof}
The set $\scone\Pi$ is called \emph{minimal support of $\Pi$} and intuitively
correspond to the ``minimal'' set of finite executions needed to completely 
characterize the behavior captured by $\Pi$ and the complete paths it induces.
Any other path of $\Pi$ is therefore redundant (\cf~Example \ref{ex:class-reachability}).

The idea of which complete paths are distinguishable and then ``measurable''
(\ie~that can be given weight) is captured precisely by the notion of 
$\sigma$-algebra. In fact, the set of all cones $\Gamma$ (together
with the emptyset) induce a $\sigma$-algebra, as they form a
\emph{semiring of sets} (in the sense of \cite{zaanen1958:integration}).
\begin{lemma}
\label{lem:cpath-semiring}
The set $\Gamma\cup\{\emptyset\}$ is a semiring of sets
and uniquely induces a $\sigma$-algebra over $\mrm{CPaths}(x)$.
\end{lemma}
\begin{proof}[Proof (Sketch)]
$\Gamma\cup\{\emptyset\}$ %
is closed under finite intersections since cones are always either disjoint or
one the subset of the other. Set difference follows from the existence
of minimal supports.
\end{proof}

As discussed before, in general the weight of
$\Pi\subseteq\mrm{FPaths}(x)$ cannot be defined as the sum of the
weights of its elements, due to redundancies.  However, what we are
really interested in is the \emph{unique} set of behaviors
described by $\Pi$, \ie~the complete paths it subsumes.  
Therefore we first extend $\rho$ to \emph{minimal} $\Pi$, as follows:
\[
 \rho(\Pi) \defeq \sum_{\pi \in \Pi}\rho(\pi)  \quad \mbox{ for } \Pi \mbox{ minimal.}
\]
then, for all $\Pi$, we simply take
\[
\rho(\Pi) \defeq \rho(\scone\Pi)\text.
\]
Because $\Pi$ can be countably infinite, semiring addition 
\emph{has to} support countable
additions over these sets (\cf~Remark~\ref{rem:wlts-strong}).

\subsection{Well-behaved semirings}\label{sec:good-semirings}

\begin{definition}\label{def:good-semiring}
Let the semiring $\mfk W$ be endowed with a preorder $\sqsubseteq$.
We call the semiring \emph{well-behaved} if, and only if, for
any two $\Pi_1$ and $\Pi_2$ the following holds:
\[\Pi_1 \subseteq \Pi_2 \Rightarrow \rho(\Pi_1) \sqsubseteq \rho(\Pi_2)\text.\]
\end{definition}
If the semiring is well-behaved then addition unit $0$ is necessarily the 
bottom of the preorder because $\rho(\emptyset) \defeq 0$.
Moreover, the semiring operations have to respect the preorder \eg:
\[ a\sqsubseteq b \Rightarrow a + c \sqsubseteq b + c\text.\]
As a direct consequence, annihilation of parallel is
avoided by the \emph{zerosumfree} property of the semiring
 \ie~the sum of weighs 
of proper transition always yield the weight of a proper 
transition where proper means different from the addition unit.

Well-behaved semirings are precisely \emph{positively (partially) ordered
semirings} and it is well known that these admit the natural preorder:
\[ a \trianglelefteq b \defiff \exists c. a + c = b\]
which is respected by the semiring operation and has $0$ as bottom.
The natural preorder is the weaker preorder rendering a semiring
positively ordered (hence well-behaved) where weaked means that
for any such preorder $\sqsubseteq$ and elements $a,b$
\[a \trianglelefteq b \implies a \sqsubseteq b\text.\]
The converse holds only when also the other order is natural.
\begin{lemma}
The natural preorder is the weaker preorder rendering the semiring
well-behaved.
\end{lemma}
Note that any idempotent semiring bares a natural preorder and hence
is well-behaved and the same holds for every semiring considered in
the examples illustrated in this paper (\cf~Section \ref{sec:weak-instances}).
For instance, some arithmetic semirings like $(\mbb R,+,0,\cdot,1)$
are not positively ordered because of negatives; moreover their are not
$\omega$-semirings (there is no limit for $1 + (-1) + 1 + (-1) \dots$).

\subsection{Weak \mfk W-bisimulation}
\label{sec:weak-wbisim}
Weak bisimulations weakens the notion of strong bisimulation
by allowing sequences of silent action before and after any observable one.
Then, we are now dealing with (suitable) paths instead of single transitions
and the states are compared on the bases of how opportune classes of states
are reached from these by means of the paths allowed (\ie~making some silent actions,
before and after an observable, if any).
Therefore, the notion of how a class state is reached and what paths can be used in
doing this is crucial in the definition of the notion of weak bisimulation.

For instance, for non-deterministic LTSs, the question of how and if a class
is reached coincides and then it suffices to find a (suitable) path leading to the 
class.  This allows weak bisimulation for non-deterministic LTSs to rely on the
reflexive and transitive closure of $\tau$-labelled transition of
a system (\cf~Definition~\ref{def:lts-weak}) to blur the distinction
between sequences of silent actions which can then be ``skipped''.
In fact, the $\tau$-closure at the base of \eqref{eq:tau-clos}
defines a new LTS over the same state space
of the previous and such that every weak bisimulation for this new system
is a weak bisimulation for the given one and vice versa.

In \cite{pbpk03:weakautom} Buchholz and Kemper extends this notion to
a class of automatons weighted over suitable semirings \ie~those
having operations commutative and idempotent (\eg~$w+w=w$).
This class includes interesting examples such as the boolean and 
bottleneck semiring (\cf~Section \ref{sec:other-semirings})
but not the semiring of non-negative real numbers and therefore
does not cover the cases of fully probabilistic systems.
Modulo some technicality connected to initial and accepting states,
their results can be extended to labelled transition systems
and holds also for LTSs weighted over suitable semirings.

Their interesting construction relies on the $\tau$-closure
of a system and it is known that this closure does not cover
the general case. For instance, it can not be applied to recover 
weak bisimulation for generative systems as demonstrated by 
Baier and Hermanns (\cf~\cite{baier98}). 
The following example gives an intuition of the issue.
\begin{example}\label{ex:class-reachability}\rm
\newcommand{\convexpath}[2]{
[   
    create hullnodes/.code={
        \global\edef\namelist{#1}
        \foreach [count=\counter] \nodename in \namelist {
            \global\edef\numberofnodes{\counter}
            \node at (\nodename) [draw=none,name=hullnode\counter] {};
        }
        \node at (hullnode\numberofnodes) [name=hullnode0,draw=none] {};
        \pgfmathtruncatemacro\lastnumber{\numberofnodes+1}
        \node at (hullnode1) [name=hullnode\lastnumber,draw=none] {};
    },
    create hullnodes
]
($(hullnode1)!#2!-90:(hullnode0)$)
\foreach [
    evaluate=\currentnode as \previousnode using \currentnode-1,
    evaluate=\currentnode as \nextnode using \currentnode+1
    ] \currentnode in {1,...,\numberofnodes} {
-- ($(hullnode\currentnode)!#2!-90:(hullnode\previousnode)$)
  let \p1 = ($(hullnode\currentnode)!#2!-90:(hullnode\previousnode) - (hullnode\currentnode)$),
    \n1 = {atan2(\x1,\y1)},
    \p2 = ($(hullnode\currentnode)!#2!90:(hullnode\nextnode) - (hullnode\currentnode)$),
    \n2 = {atan2(\x2,\y2)},
    \n{delta} = {-Mod(\n1-\n2,360)}
  in 
    {arc [start angle=\n1, delta angle=\n{delta}, radius=#2]}
}
-- cycle
}
Consider the \mfk W-LTS below.
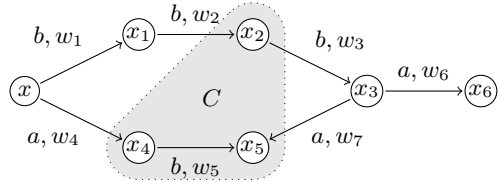
\[\begin{tikzpicture}[auto,scale=.75,font=\small,
  state/.style={circle,draw=black,fill=white,
    minimum size=11pt, inner sep=.5pt,outer sep=1pt}]
    \node[state] (n0) at (0,1) {\(x\)};
    \node[state] (n1) at (2,2) {\(x_1\)};
    \node[state] (n2) at (4,2) {\(x_2\)};
    \node[state] (n3) at (6,1) {\(x_3\)};
    \node[state] (n4) at (2,0) {\(x_4\)};
    \node[state] (n5) at (4,0) {\(x_5\)};
    \node[state] (n6) at (8,1) {\(x_6\)};
    \draw[->]  (n0) to node [pos=.65] {\(b,w_1\)} (n1);
    \draw[->]  (n1) to node [] {\(b,w_2\)} (n2);
    \draw[->]  (n2) to node [pos=.45] {\(b,w_3\)} (n3);
    \draw[->]  (n3) to node [pos=.6] {\(a,w_7\)} (n5);
    \draw[->]  (n0) to node [pos=.6,swap] {\(a,w_4\)} (n4);
    \draw[->]  (n4) to node [swap] {\(b,w_5\)} (n5);
    \draw[->]  (n3) to node [] {\(a,w_6\)} (n6);
    \begin{pgfonlayer}{background}
      \draw[black!90,dotted,fill=gray!20] \convexpath{n2,n5,n4}{16pt};
    \end{pgfonlayer}
    \node[] (c) at ($(n2)!0.45!(n4)!0.2!(n5)$) {\(C\)};
  \end{tikzpicture}
\]
There are four finite paths going from state the $x$ to the class $C$. 
Their weights are:
\begin{align*}
\rho(x \xrightarrow{b} x_1\xrightarrow{b} x_2) &= w_1\cdot w_2\\
\rho(x \xrightarrow{b} x_1\xrightarrow{b} x_2\xrightarrow{b} x_3\xrightarrow{a} x_5) 
&= w_1\cdot w_2\cdot w_3\cdot w_7\\
\rho(x \xrightarrow{a} x_4) &= w_4\\
\rho(x \xrightarrow{a} x_4 \xrightarrow{b} x_5) &= w_4\cdot w_5
\end{align*}
Let us suppose to define the weight
of the set of these paths as the sum of its elements weights and
suppose that the system is generative; then the probability of
reaching $C$ from $x$ would exceed $1$. Likewise, in the case of a
stochastic system, the rate of reaching $C$ cannot consider paths
passing through $C$ before ending in it. If we are interested in how
$C$ is reached from $x$ with actions yielding a trace in the set
$b^*ab^*$, paths $w_1\cdot w_2$ and $w_5\cdot w_6$ are ruled out
because the first has a different trace and the second reaches $C$
before it ends.
\end{example}

Then, given a set of traces $T$, a state $x$ and a class of states $C$,
the set of finite paths of the given transition system reaching $C$
from $x$ with trace in $T$ that should be considered is:
\[\Lbag x,T,C \Rbag\defeq
  \Biggl\{\pi\;\Bigg|\;
  \parbox{6.5cm}{
    ${\pi\in \mrm{FPaths}(x)}$,
    ${\mrm{last}(\pi)\in C}$, ${\mrm{trace}(\pi) \in T}$,
    ${\forall \pi'\preceq\pi : \mrm{trace}(\pi') \in T \Rightarrow 
    \mrm{last}(\pi') \notin C}$
  }
\Biggr\}
\]
since these are all and only the finite executions of the system starting
going from $x$ to $C$ with trace in $T$ and never passing through $C$ except
for their last state. Redundancies highlighted in the example above
are ruled out since no execution path in this set is the prefix of an other
in the same set. In particular $\Lbag x,T,C \Rbag$ is the minimal support of the set of all
finite paths reaching $C$ from $x$ with trace in $T$:
\[\Lbag x,T,C \Rbag = \scone{\{\pi\mid\pi\mrm{FPaths}(x),\mrm{last}(\pi)\in C, \mrm{trace}(\pi) \in T\}}\text.\]
Therefore, weight functions can be consistently extended to these sets 
by point-wise sums:
\[
\rho(\Lbag x,T,C \Rbag) = \sum_{\pi\in\Lbag x,T,C \Rbag} \rho(\pi)\text.
\]
The sum is at most countable since $\mrm{FPaths}(x)$ so is
$\mrm{FPaths}(x)$ and $\Lbag x,T,C \Rbag\subseteq\mrm{FPaths}(x)$.
Then, the addition operation of the semiring will support countable
sums as discussed in Remark~\ref{rem:wlts-strong}.

When clear from the context, we may omit the bag brackets from $\rho(\Lbag x,T,C \Rbag)$.

We are now ready to state the notion of weak bisimulation of a labelled transition 
system weighted over any semiring admitting sums over (not necessarily every)
countable family of weights. 
The notion we propose relies on the weights of paths reaching every class
in the relation but making at most one observable and hence the importance
of defining sets of paths reaching a class consistently.

\begin{definition}[Weak \mfk W-bisimulation]
\label{def:wlts-weak}
Let $(X,A+\{\tau\},\rho)$ be a LTS weighted over the semiring
\mfk W.  A \emph{weak \mfk W-bisimulation} is an equivalence relation
$R$ on $X$ such that for all $x,x'\in X$, $(x,x') \in R$ implies that
for each label $a\in A$ and each equivalence class $C$
of $R$:
\begin{gather*}
\rho(x,\tau^*a\tau^*,C) = \rho(x,\tau^*a\tau^*,C)\\
\rho(x,\tau^*,C) = \rho(x,\tau^*,C)
\text.\end{gather*}
  States $x$ and $x'$ are said to be \emph{weak \mfk W-bisimilar} 
  (or just weak bisimilar), written  $x \approx_\mfk W x'$, if there exists a 
  weak \mfk W-bisimulation $R$ such that $x R x'$.
\end{definition}

The approach we propose applies to other behavioural equivalences.
For instance, delay bisimulation can be recovered for
WLTSs by simply considering in the above definition of weak bisimulations
sets of paths of the sort of $\Lbag x,\tau^*,C \Rbag$ and $\Lbag x,\tau^*a,C \Rbag$.
The notion of branching bisimulation relies on paths with the same traces 
of those considered for defining weak bisimulation but with some additional
constraint on the intermediate states. In particular, the states right
before the observable $a$ have to be in the same equivalence class and 
likewise the states right after it. Definition~\ref{def:wlts-weak}
is readily adapted to branching bisimulation by considering these particular
subsets of $\Lbag x,\tau^*a\tau^*,C \Rbag$.

\section{Examples of weak \mfk W bisimulation}
\label{sec:weak-instances}

In this Section we instantiate Definition~\ref{def:wlts-weak} to
the systems introduced in Section~\ref{sec:wlts} as instances of
LTSs weighted over commutative $\omega$-monoids.

\subsection{Non-deterministic systems} 
Let us recall the usual definition of weak bisimulation for LTS \cite{milner:cc}.
\begin{definition}[Weak non-deterministic bisimulation]
\label{def:lts-weak}
  An equivalence relation $R\subseteq X\times X$ is a \emph{weak (non-deterministic)
  bisimulation on $(X,A+\{\tau\},\rightarrow)$} iff for each $(x,x') \in R$, 
  label $\alpha \in A+\{\tau\}$ and equivalence class $C \in X/R$:
  \begin{equation}\label{eq:tau-clos}
    \exists y \in C. x\xRightarrow{\alpha}y \iff \exists y' \in C. x'\xRightarrow{\alpha}y'
  \end{equation}
  where $\mathop{\Rightarrow} \subseteq X\times (A\uplus\{\tau\}) \times X$ is the
  well-known $\tau$-reflexive and $\tau$-transitive closure of the transition relation 
  $\rightarrow$.
  Two states $x$ and $x'$ are said \emph{weak bisimilar} iff there exists 
  a weak non-deterministic bisimulation relation $\approx_n$ such that $x\approx_n y$.
\end{definition}
Clearly, a weak bisimulation is a relation on states
induced by a strong bisimulation of a suitable LTS with the same states and actions.
In particular, weak bisimulations for $(X,A+\{\tau\},\rightarrow)$  are strong
bisimulations for $(X,A+\{\tau\},\Rightarrow)$ and viceversa.
The transition system $(X,A+\{\tau\},\Rightarrow)$ is sometimes referred as 
saturated or weak (\eg~in \cite{jensen:thesis}). This observation is at the base of
some algorithmic and coalgebraic approaches to weak non-deterministic bisimulations
(\cf~Section~\ref{sec:cw-algorithm} and Section~\ref{sec:cat-view} respectively).

Section~\ref{sec:lts-as-wlts} illustrated that non-deterministic LTSs are
$\mfk 2$-WLTSs. The commutative monoid \mfk 2 is part of the 
\emph{boolean semiring} of logical values under disjunction and conjunction
$(\{\ltrue,\lfalse\},\lor,\lfalse,\land,\ltrue)$ which we shall also denote
as \mfk 2. Then, by straightforward application of the definitions, the notions of 
weak non-deterministic bisimulation and  weak \mfk 2-bisimulation coincide.
\begin{proposition}%
Definition \ref{def:lts-weak} is equivalent to Definition
\ref{def:wlts-weak} with $\mfk W = \mfk 2$.
\end{proposition}
It easy to check that a similar correspondence holds for branching and delay 
bisimulations.

\subsection{Probabilistic systems}
In the definition of weak bisimulation for fully probabilistic systems
we are interested in the probability of reach a class of states.  This
aspect is present also in the case of strong bisimulation, but things
become more complex for weak equivalences due to silent actions and
multi-step executions.  Moreover, $\sigma$-additivity is no longer
available since the probability of reaching a class of states is not
the sum of the probabilities of reaching every single state in that
class.  (On the contrary, a class is reachable if any of its state is
so which is the property we are interested in when dealing with
non-deterministic systems.)

Weak bisimulation for fully probabilistic systems was introduced by
Baier and Hermanns in \cite{baier97:cav,baier98}. Here we
recall briefly their definition; we refer the reader to
\emph{loc. cit.}  for a detailed presentation.

\begin{definition}[Weak probabilistic bisimilarity {\cite{baier97:cav,baier98}}]
\label{def:plts-weak}
  Given a fully probabilistic system $(X,A+\{\tau\},\mrm P)$,
  an equivalence relation $R$ on $X$ is a \emph{weak (probabilistic)  
  bisimulation} iff for $(x,x') \in R$,
  for any $a \in A$ and any equivalence class $C \in X/R$:
  \begin{gather*}
    \mrm{Prob}(x,\tau^*a\tau^*,C) = \mrm{Prob}(x',\tau^*a\tau^*,C) \\
    \mrm{Prob}(x,\tau^*,C) = \mrm{Prob}(x',\tau^*,C)
  \text.\end{gather*}
  Two states $x$ and $x'$ are said \emph{weak bisimilar} iff there exists 
  a weak probabilistic bisimulation relation $\approx_p$ such that $x\approx_p y$.
\end{definition}

The function \mrm{Prob} is the extension over finite execution paths of the
unique probability measure induced by $\mrm{P}$ over the $\sigma$-field of
the basic cylinders of complete paths. 

\begin{proposition}%
Definition \ref{def:lts-weak} is equivalent to Definition
\ref{def:wlts-weak} with $\mfk W =  (\mathbb R_0^+,+,0,\cdot,1)$.
\end{proposition}
The function $\mrm{P}$ is a weight function such that $\mrm{P}(x,\hole,\hole)$
is a probability measure (or the constantly $0$ measure) which extends to the unique 
$\sigma$-algebra on $\mrm{CPaths}(x)$ (Lemma~\ref{lem:cpath-semiring}). This defines 
precisely $Prob$. In particular, for any $x \in X$ and $\Pi \subseteq \mrm{FPaths}(x)$ 
$Prob(\Pi) = Prob(\scone{\Pi}) = \mrm{P}(\scone{\Pi}) = \mrm{P}(\Pi)$ where $\mrm{P}$
is seen as the weigh function of a $\mbb R^+_0$-LTS. 

\subsection{Stochastic systems}
As we have seen in Section~\ref{sec:rts-wlts}, stochastic transition
systems can be captured as WLTSs over $(\overline{\mbb R}^+_0,+,0)$ by describing
the exponential time distributions of a CTMC by their rates
\cite{klinS08}.  Unfortunately, this does not extend to paths because
the sequential composition of two exponential distributions does not
yield an exponential distribution, and hence it can not be represented
by an element of $\overline{\mbb R}^+_0$.  Moreover, there are stochastic systems
(\eg~TIPP \cite{gotz93:tipp}, SPADES \cite{ak2005:spades}) whose
transition times follow generic probability distributions.

To overcome this shortcoming, in this Section we introduce a semiring of
weights called \emph{stochastic variables} which allows to express
stochastic transition system with generic distributions as WLTSs. Then
the results of this theory can be readily applied to define various
behavioural equivalences, ranging from strong bisimulation to trace
equivalence, for all these kind of systems. In particular, we define
\emph{weak stochastic bisimulation} 
by instantiating Definition~\ref{def:wlts-weak} on the semiring of 
stochastic variables. 

The carrier of the semiring structure we are defining is the set \mbb T
of \emph{transition-time random variables} \ie~random variables on 
the nonnegative real numbers (closed with infinity) which describes 
the nonnegative part of the line of time.

Given two (possibly dependent) random variables $X$ and $Y$
from \mbb T, let $\min(X,Y)$ be the minimum random variable
yielding the minimum between $X$ and $Y$. If the variables $X$ and
$Y$ characterize the time required by two transitions then their 
combined effect is defined by the stochastic race between the
two transitions; a race that is ``won'' by the transition completed earlier
and hence the minimum. For instance, given two stochastic transitions 
${x \xrightarrow{X} x'}$ and ${y\xrightarrow{Y}y'}$ the transition time
for their ``combination'' going from $\{x,y\}$ to $\{x',y'\}$
is characterized by the random variable $\min(X,Y)$ \ie~the overall
time is given by the first transition to be completed on the specific run.

Minimum random variables defines the operation $\min$ over \mbb T with a 
constantly $+\infty$ continuous random variable $\mcl T_{+\infty}$
(its density is the Dirac delta function $\delta_{+\infty}$) as the unit.
Random variables of the sort of $\mcl T_{+\infty}$ are self-independent
and since they always always yield $+\infty$ we shall make no distinction
between them and refer to \emph{the} $\mcl T_{+\infty}$ random variable.
In general, time-transition variables do not have to be self-independent since
the events they describe usually depends on themselves. Intuitively, 
it is like racing against ourself \ie~we are the only racer
and therefore $\min(X,X) = X$. Formally: 
\[
\mbb P(\min(X,X) > t) 
= \mbb P(X > t \cap X > t)
= \mbb P(X > t)\cdot \mbb P(X > t\mid X > t) = \mbb P(X > t)\text.
\]

Let $X$ and $Y$ be two continuous random variables from \mbb T with probability
density functions $f_X$ and $f_Y$ respectively. The density $f_{\min(X,Y)}$
describing $\min(X,Y)$ is:
\begin{align*}
f_{\min(X,Y)}(z) = f_X(z) + f_Y(z) - f_{X,Y}(z,z)\text.
\end{align*}
When $X$ and $Y$ are independent (but not necessarily \emph{i.i.d.}) 
$f_{\min(X,Y)}$ can be simplified as:
\[
f_{\min(X,Y)}(z) = 
f_X(z)\cdot\int_z^{+\infty} f_Y(y)\dif y + 
f_Y(z)\cdot\int_z^{+\infty} f_X(x)\dif x
\text.
\]
Intuitively, the
likelihood that one variable is the minimum must be ``weighted'' by
the probability that the other one is not.
In particular, for independent exponentially distributed variables
$X$ and $Y$, $\min(X,Y)$ is exponentially distributed and its
rate is the sum of the rates of the negative exponentials characterizing
$X$ and $Y$. Therefore, the commutative monoid 
$(\mbb T,\min,\mcl T_{+\infty})$ faithfully generalizes the monoid 
$(\overline{\mbb R}^+_0,+,0)$ used in Section~\ref{sec:rts-wlts} to capture 
CTMCs as WLTSs

During the execution of a given path, the time of every transition in the sequence
sums to the overall time. Therefore, the transition time for \eg~$x \xrightarrow{X} y \xrightarrow{Y} z$ is characterized by the random variable 
$X+Y$ sum of the variable characterizing the single transitions composing the path.
Sum and the constantly 0 continuous variable $\mcl T_0$ 
define a commutative monoid over \mbb T. The operation has to be commutative 
because the order a path imposes to its steps does not change the total time of 
execution.

Let $X$ and $Y$ be two continuous random variables from \mbb T with probability
density functions $f_X$ and $f_Y$ respectively. The probability density 
function $f_{X+Y}$ is:
\[
  f_{X+Y}(t) = \int_0^t f_{X,Y}(s,t-s)\dif
  s
\]
and, if $X$ and $Y$ are independent (but not necessarily \emph{i.i.d.}), $f_{X+Y}$ is the convolution:
\[
  f_{X+Y}(t) = \int_0^t f_X(s)\cdot f_Y(t-s)\dif
  s\text.
\]

It is easy to check that sum distributes over
minimum:
\[X+\min(Y,Z) = \min(X+Y,X+Z)\]
by taking advantage of the latter operation being idempotent.
Then, because of sum being commutative, 
left distributivity implies right one (and vice versa).
Thus $\mfk S \defeq (\mbb T,
\min,\mcl T_{+\infty},+,\mcl T_{0})$ is a (commutative and idempotent)
semiring and stochastic systems can be read as \mfk S-LTS.
This induces immediately a strong bisimulation (by instantiating
Definition~\ref{def:wlts-strong}) which corresponds to strong stochastic
bisimulations on rated LTS (Definition \ref{def:rts-bisim}).
Moreover, following Definition \ref{def:wlts-weak}, we can readily
define the weak stochastic bisimulation as the weak \mfk
S-bisimulation.

In literature there are some (specific and ad hoc) notions of weak
bisimilarity for stochastic systems.  The closest to our is the one
given by Bernardo et al.~for CTMCs extended with passive rates and
\emph{instantaneous} actions \cite{mb2007:ictcs,mb2012:qapl}. Their
definition is finer than our weak \mfk S-bisimulation since they allow
to merge silent actions only when these are instantaneous and hence
unobservable also \wrt~the time.  Instead, in our definition sequences
of $\tau$ actions are equivalent as long as their overall ``rates'' are
the same (note that in general, the convolution of exponentially distributed
random variables is no longer exponentially distributed but an hyper-exponential).
In \cite{mb2012:qapl}, Bernardo et al. relaxed the definition given in 
\cite{mb2007:ictcs} to account also for non-instantaneous $\tau$-transitions.
However, to retain exponentially distributed variables, they approximate
hyper-exponentials with exponentials with the same average. This approach
allows them to obtain a saturated system that still is a CTMC 
but loosing precision since, in general, the average is the only momentum preserved
during the operation. On the opposite, our approach does not introduce
any approximation.

In \cite{ln2005:ictac} L{\'o}pez and N{\'u}{\~n}ez proposed a definition of weak
bisimulation for stochastic transition systems with generic distributions. Their 
(rather involved) definition is a refinement of the notion they previously 
proposed in \cite{ln2004:fac} and relies on the reflexive and transitive closure
of silent transitions. However, their definition of strong 
bisimulation does not correspond to the results from the theory of
WLTSs, so neither the weak one does.

\subsection{Other examples}
\label{sec:other-semirings}
The definition of weak \mfk W-bisimulation applies to many other
situations. In the following we briefly illustrate some interesting cases.

\paragraph{Tropical and arctic semirings}
These semirings are used very often in optimization problems,
especially for task scheduling and routing problems. Some examples
are: $(\overline{\mbb R},\min,+\infty,+,0)$; $(\overline{\mbb R},\max,-\infty,+,0)$;
$(\overline{\mbb R},\min,+\infty,\max,-\infty)$.

In these contexts, weak bisimulation would allow to abstract from
``unobservable'' tasks \eg~internal tasks and treat a cluster of
machines as a single one, reducing the complexity of the problem.

\paragraph{Truncation semiring} 
$(\{0,\dots,k\},\max,0,\min\{\hole+\hole,k\},k)$.  It is variant of the
above ones, and it is used to reason ``up-to'' a threshold $k$.  A
weak bisimulation for this semiring allows us to abstract from how the
threshold is violated, but only if this happens.

\paragraph{Probabilistic semiring}
Another semiring used for reasoning about probabilistic events is
$([0,1],\max,0,\cdot,1)$. This is used to model the maximum likelihood
of events, \eg~for troubleshooting, diagnosis, failure forecasts,
worse cases, etc.  A weak bisimulation on this semiring allows to
abstract from ``unlikely'' events, focusing on the most likely ones.

\paragraph{Formal languages}
A well-known semiring is that of formal languages over a given alphabet
$(\wp(\Sigma^*),\cup,\emptyset,\circ,\varepsilon)$.  Here, a weak
bisimulation is a kind of \emph{determinization} \wrt~to words assigned to 
$\tau$ transitions. %

\section{A parametric algorithm for computing weak \mfk
  W-bisim\-u\-la\-tions}\label{sec:cw-algorithm}
In this section we present an algorithm for computing weak \mfk
W-bisimulation equivalence classes which is parametric in the semiring
structure \mfk W.  Being parametrized, the same algorithm can be used
in the mechanized verification and analysis of many kinds of systems.
This kind of algorithms is often called \emph{universal} since they do
not depend on any particular numerical domain nor its machine
representation.
In particular, algorithms parametric over a semiring structure have
been successfully applied to other problems of computer science,
especially in the field of system analysis and optimization (\cf~\cite{lmrs2011}).

The algorithm we present is a variation of the well-known
Kanellakis-Smolka's algorithm for deciding strong non-deterministic
bisimulation \cite{ks1990:ic}.
Our solution is based on the same refinement technique used for the
\emph{coarsest stable partition}, but instead of ``strong''
transitions in the original system we consider ``weakened'' or
saturated ones.  The idea of deciding weak bisimulation by computing
the strong bisimulation equivalence classes for the saturated version
of the system has been previously and successfully used \eg~for
non-deterministic or probabilistic weak bisimulations \cite{baier98}.
The resulting complexity is basically that of the coarsest stable
partition problem plus that introduced by the construction of the
saturated transitions.  The last factor depends on the properties and
kind of the system and, in our case, on the properties of the semiring
\mfk W (the algorithm and its complexity will be discussed with more
detail in Section~\ref{sec:cw-algo-and-complexity}).

Before outlining the general idea of the algorithm let us introduce
some notation.  For a finite set $X$ we denote by $\mcl{X}$ a
partition of it \ie~a set of pairwise disjoint sets $B_0,\dots,B_n$
covering $X$:
\[
  X = \biguplus \mcl X = \biguplus \{B_0,\dots,B_n\}\text.
\]
We shall refer to the elements of the partition \mcl X as \emph{blocks} or 
\emph{classes} since every partition induces an equivalence relation 
$\bigcup_{B\in\mcl X} B\times B$ on $X$ and vice versa. 

Given a finite \mfk W-LTS $(X,A+\{\tau\},\rho)$ the general idea
for deciding weak \mfk W-bisimulation by partition refinement is 
to start with a partition of the states $\mcl{X}_0$
coarser than the weak bisimilarity relation \eg~$\{X\}$ and then successively
refine the partition with the help of a \emph{splitter} (\ie~a witness
that the partition is not stable \wrt~the transitions). This process
eventually yields a partition $\mcl{X}_k$ 
being the set of equivalence classes of the weak bisimilarity. 
A splitter of a partition $\mcl X$ is a pair made of an action and a class of $\mcl X$
that violates the condition for $\mcl{X}$ to be a weak bisimulation. 
Reworded, a pair $\langle \alpha,C\rangle \in (A+\{\tau\})\times\mcl X$ is a splitter
for $\mcl X$ if, and only if, there exist $B \in \mcl X$ and $x,y\in B$ such that:
\begin{equation}\label{eq:cw-split-check}
  \rho(x,\hat\alpha,C) \neq \rho(y,\hat\alpha,C)
\end{equation}
where $\hat\alpha$ is a short hand for the sets of traces $\tau^*$ 
and $\tau^*a\tau^*$ when $\alpha=\tau$ and $\alpha = a \in A$  respectively.
Then $\mcl X_{i+1}$ is obtained from $\mcl X_i$ splitting
every\footnote{In Kanellakis and Smolka's algorithm, only the block
  $B$ is split but in our case we need to evaluate every block anyway
  because of saturation, \cf~Section~\ref{sec:cw-saturation}.} $B \in
\mcl X_i$ accordingly to the selected splitter $\langle
\alpha,C\rangle$.
\begin{equation}
  \label{eq:cw-refine}
  \mcl X_{i+1} \defeq \bigcup\left\{
  B/{\underset{\scriptscriptstyle{\alpha,C}}{\approx}}\mid B \in \mcl X_i\right\}
\end{equation}
where $\underset{\scriptscriptstyle{\alpha,C}}{\approx}$ is the equivalence relation
on states induced by the splitter and such that:
\[
  x \underset{\scriptscriptstyle{\alpha,C}}{\approx} y 
  \defiff\rho(x,\hat\alpha,C) = \rho(y,\hat\alpha,C)
\text.
\]
Note that the block $B$ can be split in more than two parts (which is the case 
of non-deterministic systems) since splitting depends on weights of outgoing 
weak transitions.

\subsection{Computing weak transitions}
\label{sec:cw-saturation}

The algorithm outlined above follows the classical approach to the
coarsest stable partition problem where stability is given in terms of
weak weighted transitions like $\Lbag x,\tau^*,C \Rbag$ (and in general weighted
sets of paths \eg~$\Lbag x,T,C \Rbag$) but nothing is assumed on how these
values are computed.  In this section, we show how weights of weak
transitions can be obtained as solutions of systems of linear
equations over the semiring \mfk W.  Clearly, for some specific cases
and sets of paths, there may be more efficient \emph{ad-hoc} technique
(\eg~saturated transitions can be precomputed
for non-deterministic LTSs)
however the linear system at the core of our algorithm is a general
and flexible solution which can be readily adapted to other
observational equivalences (\cf~Example~\ref{ex:delay-bisim}).

Let $C$ be a class. For every $x \in X$ and 
$\alpha \in A+\{\tau\}$ let $x_\alpha$ be a variable with domain 
the semiring carrier. Intuitively, once solved, these will represent:
\[
x_\tau = \rho(\Lbag x,\tau^*,C\Rbag) \qquad
x_a = \rho(\Lbag x,\tau^*a\tau^*,C\Rbag)
\]
The linear system is given by the equation families
\eqref{eq:cw-ls-in-class}, \eqref{eq:cw-ls-only-tau} and \eqref{eq:cw-ls-action}
which capture exactly the finite paths yielding the cones covering weak transitions.
\begin{alignat}{3}
\label{eq:cw-ls-in-class}
x_\tau &= 1 
        && \quad\mbox{ for } x \in C \\
\label{eq:cw-ls-only-tau}
x_\tau &= {\sum_{y \in X} \rho(x\xrightarrow{\tau}y)\cdot y_\tau} 
       && \quad\mbox{ for } x \notin C \\
\label{eq:cw-ls-action}
x_a    &= {\sum_{y \in X} \rho(x\xrightarrow{a}y) \cdot y_\tau}  + 
         {\sum_{y \in X} \rho(x\xrightarrow{\tau}y) \cdot y_a}
       &&
\end{alignat}
The system is given as a whole but it can
be split in smaller sub-systems improving the efficiency of the resolution
process. In fact, unknowns like $x_a$ depend only on those indexed by $\tau$
or $a$ and unknowns like $x_\tau$ depend only on those indexed by $\tau$.
Hence instead of a system of $|A+\{\tau\}|\cdot|X|$ equations and unknowns,
we obtain $|A+\{\tau\}|$ systems of $|X|$ equations and unknowns by
first solving the sub-system for $x_\tau$ and then a separate
sub-system of each action $a \in A$ (where $x_\tau$ are now constant).

\begin{example}[Delay bisimulation]
\label{ex:delay-bisim}
Delay bisimulation is defined at the general level of WLTSs
simply by replacing $\Lbag x,\tau^*a\tau^*,C\Rbag$ with $\Lbag x,\tau^*a,C\Rbag$
in Definition~\ref{def:wlts-weak}. Then, delay bisimulation
equivalence classes can be computed with the same algorithm simply by
changing the saturation part at its core. Weights
of sets like $\Lbag x,\tau^*a,C\Rbag$ are computed as the solution to
the linear equation system:
\[
x_a = {\sum_{y \in X} \rho(x\xrightarrow{\tau}y)} \cdot y_a+ 
      {\sum_{y \in C} \rho(x\xrightarrow{a}y)}
\text.\]
\end{example}

\subsubsection{Solvability}
Decidability of the algorithm depends on the solvability the equation
system at its core. In particular, on the existence and uniqueness of the
solution. In section we prove that this holds for every positively
ordered $\omega$-semiring. The results can be extended to
$\sigma$-semirings provided that their $\sigma$-algebra covers
the countable families used by Theorem~\ref{thm:uniq-solution}.

The linear equation systems under consideration bare a special form:
they have exactly the same number of equations and unknowns
(say $n$)
and every unknown appears alone on the left side
of exactly one equation. Therefore, these systems are
defining an operator 
\begin{equation}\label{eq:sys-matrix}
F(x) = M\times x + b
\end{equation}
over the space of $n$-dimensional vectors $W^n$
where $M$ and $b$ are a $n$-dimensional matrix
and vector respectively defined by the equations of the
system.
Then, the solutions of the system are precisely 
the fix-points of the operator $F$ and since the number of 
equations and unknowns is the same, if $F$ has a fix-point, it is unique.

Let the semiring $\mfk W$ be positively ordered. 
These semirings admit a natural preorder $\trianglelefteq$ which subsumes 
any preorder $\sqsubseteq$ respecting the structure of the semiring;
hence we restrict ourselves to the former.
The point-wise extension of $\trianglelefteq$ to $n$-dimensional vectors
defines the partial order with bottom $(W^n,\dot\trianglelefteq,0^n)$;
suprema are lifted pointwise from $(W,\trianglelefteq,0)$
where are sum-defined. Therefore, $\omega$-chains suprema exists
only under the assumption of addition over at most countable families and viceversa.
\begin{lemma}\label{lem:cpo}
$(W^n,\dot\trianglelefteq,0^n)$ is $\omega$-complete
iff $\mfk W$ admits countable sums.
\end{lemma}

The operator $F$ manipulates its arguments only by additions
and constant multiplications which respect the natural order.
Thus $F$ is monotone with respect to $\dot\trianglelefteq$.
Moreover, $F$ preserves $\omega$-chains suprema (and in general
$\omega$-families) because suprema for $\trianglelefteq$ are
defined by means of additions and the order is lifted point-wise.
\begin{lemma}\label{lem:scott}
The operator $F$ over $(W^n,\dot\trianglelefteq,0^n)$ is Scott-continuous.
\end{lemma}

Finally, we can state the main result of this Section from which decidability
follows as a corollary.
\begin{theorem}\label{thm:uniq-solution}
Systems in the form of \eqref{eq:sys-matrix}
have unique solutions if the underlying semiring is
well-behaved and $\omega$-complete.
\end{theorem}
\begin{proof}
  By Lemma \ref{lem:cpo}, Lemma \ref{lem:scott} and Kleene Fix-point
  Theorem $F$ has a least fix point.  Because the linear equation
  system has the same number of equations and unknowns, this solution
  is unique.
\end{proof}

The linear equation systems defined by
the equation families \eqref{eq:cw-ls-in-class}, \eqref{eq:cw-ls-only-tau} and \eqref{eq:cw-ls-action} have exactly one solution and hence the algorithm
proposed is decidable. Moreover this holds also for any
behavioural equivalence whose saturation can be expressed in a similar
way \eg~delay bisimulation (\cf~Example~\ref{ex:delay-bisim}).

\subsubsection{Adequacy}

If $x\in C$, then the empty execution $\varepsilon$ is the 
only element of the set $\Lbag x,\tau^*,C \Rbag$,
(by definition of reachability) $\rho(\varepsilon)$ is the value
of the $0$-fold multiplication \ie~the unit $1$. 
This case falls under \eqref{eq:cw-ls-in-class} and 
hence $x_\tau$ is $\rho(x,\tau^*,C)$ when $x \in C$.

On the other hand, if $x\notin C$, then every path reaching $C$ from $x$ 
needs to have length strictly greater than $0$; reworded, it starts with a transition
${x\xrightarrow{\tau}y}$ and from $y$ heads towards $C$. 
The weight of $\Lbag x,\tau\tau^*,C \Rbag$ is the sum of the weights
of its paths which are themselves the ordered multiplication of their
steps. Then by grouping paths by their second state the remaining parts
are exactly the paths in the set $\Lbag y,\tau^*,C \Rbag$.
Then we obtain the unfolding
\[\rho(\Lbag x,\tau\tau^*,C \Rbag) = \sum_{y\in X}\rho(x\xrightarrow{\tau}y)\cdot\rho(\Lbag y,\tau^*,C \Rbag)\]
which recursively defines the weight of these sets as the unfolding of 
executions. In particular, the base case is precisely \eqref{eq:cw-ls-only-tau}
and the inductive one is \eqref{eq:cw-ls-in-class}.

Every path in the set $\Lbag x,\tau^*a\tau^*,C \Rbag$ contains exactly
one transition labelled by the action $a$ and hence it has a transition,
to some state $y$, and is labelled with either $a$ or $\tau$. 
In the first case , the observable $a$ is consumed and remaining path 
is necessarily in the set $\Lbag y,\tau^*,C \Rbag$ covered above.
In the second case, the only observable of the path has not been consumed yet
and thus the remaining part of the path should be in the set 
$\Lbag y,\tau^*a\tau^*,C \Rbag$
completing the case for \eqref{eq:cw-ls-action}.

\begin{proposition}
Let $\mfk W$ be a positively ordered $\omega$-semiring.
For any $C$, $\alpha$ and $x$, solutions for
\eqref{eq:cw-ls-in-class}, \eqref{eq:cw-ls-only-tau} and \eqref{eq:cw-ls-action}
are exactly the weights of $\Lbag x,\hat\alpha,C \Rbag$.
\end{proposition}

\subsection{The algorithm and its complexity}
\label{sec:cw-algo-and-complexity}

\begin{figure}[t]
\begin{algorithmic}[1]
\STATE $\mcl{X} \leftarrow \{X\}$
\STATE $\mcl{X}'\leftarrow \emptyset$
\REPEAT
  \STATE $changed \leftarrow\FALSE$
  \STATE $\mcl{X}''\leftarrow \mcl{X}$
  \FORALL{$C \in \mcl{X}\setminus\mcl{X}'$}
    \FORALL{$\alpha \in A+\{\tau\}$}
      \IF{$\langle\alpha,C\rangle$ is a split}
        \vspace{3pt}
        \STATE $\mcl{X} \leftarrow \bigcup\{
          B/{\underset{\scriptscriptstyle{\alpha,C}}{\approx}}\mid 
          B \in \mcl X\}$
        \STATE $changed \leftarrow\TRUE$
      \ENDIF
    \ENDFOR
  \ENDFOR
  \STATE $\mcl{X}'\leftarrow \mcl{X}''$
\UNTIL{\NOT$changed$}
\RETURN \mcl X
\end{algorithmic}
\hrule{}
\caption{The algorithm for weak \mfk W-bisimulation.}
\label{fig:cw-algo}
\end{figure}

In this section we describe the algorithm and study its worst
case complexity. The algorithm and the resulting analysis follow the structure
of the Kanellakis-Smolka's result. However, some assumptions available
in the case of strong bisimulation for non-deterministic systems 
are not available in this settings. For instance, transitions have to be computed
on the fly. Moreover, like many other algorithms parametrized over
semirings, no hypotheses are made over the numerical domain nor
over its machine representation. As a consequence, we can not assume
constant-time random access data-structures or linearly order the
elements of the semiring. However, since many practical semirings
admit total-orderings and efficient data structures, we will describe
also this second case providing a more efficient version of the 
algorithm for the general case.

The first algorithm we propose is reported in Figure~\ref{fig:cw-algo}.
Given a finite \mfk W-LTS $(X,A+\{\tau\},\rho)$ as input, it returns a
partition \mcl{X} of $X$ inducing a weak \mfk W-bisimulation for the
system.

The partition \mcl X is initially assumed to have the set of states $X$ 
as its only block and corresponds
to the assumption of the largest possible equivalence relation on $X$ being also a
weak bisimulation. In general, any partition coarser than some weak bisimulation 
would be a suitable initial partition.

The purpose of the two auxiliary partitions $\mcl X'$ and $\mcl X''$ 
is to keep track of which classes were added to \mcl X during the
previous iteration of the repeat-until loop and thus avoiding to reuse a split
candidate. We used these additional partitions for readability but the same result
may be achieved, for instance, having two colours distinguishing blocks already
checked. Moreover, $\mcl X'$ and $\mcl X''$ make the flag $changed$ redundant.

The algorithms iterates over each split candidate
$\langle \alpha,C\rangle$ and tries to split the partition by checking
whatever \eqref{eq:cw-split-check} holds. If the partition ``survives''
to every split test then it is stable and in particular it describes
a weak $\mfk W$-bisimulation relation. The saturated transitions
required to test $\langle \alpha,C\rangle$ are computed
by solving the linear equations system described before.
Overall, we have to solve ${|A|+1}$ systems of $|X|$ linear equations and unknowns
for each $C$. 

The complexity of solving these systems depends on the underling
semiring structure. For instance, solving a system over the semiring 
of non-negative real numbers is in \mrm P \cite{aho74:algobook},
whereas solving a system over the tropical (resp. arctic) semiring 
is in $\mrm{NP} \cap \mrm{coNP}$ (\cf~\cite{gp12:corr}). Since
the algorithm is parametrized by the semiring, its complexity will 
be parametrized by the one introduced by the solution of these
linear equation systems. Therefore we shall denote by
$\mcl{L}_\mfk  W(n)$ the complexity of solving a system of $n$ linear
equations in $n$ variables over \mfk W.

\begin{remark}
\label{rem:cw-efficient-split}
The complexity of the split test can be made preciser since we are not
solving a general linear system, but a specific sub-class of these.
For instance, solving a linear system over the boolean semiring is
NP-Complete in general, whereas we are interested in a specific
subclass of those encoding a reachability problem over a directed
graph which is in P.
\end{remark}

Let $n$ and $m$ denote the cardinality of states and
labels respectively.  For each block $C$ used to generate splits,
there are exactly $m$ candidates requiring to solve $m$ split tests
and perform at most $m$ updates to \mcl X.  Splits can be thought
describing a tree whose nodes are the various blocks encountered by
the algorithm during its execution and whose leaves are exactly the
elements of the final partition. Because the cardinality of \mcl X is
bound by $n$, the algorithm can encounter at most $\mcl O(n)$ blocks
during its entire execution and hence it performs at most $\mcl O(n)$
updates of \mcl X (which happens when splits describe a perfect tree with
$n$ leaves).  Therefore, in the worst case, the algorithm does
$\mcl{O}(nm)$ split tests and $\mcl{O}(n)$ partition refinements.
Partition refinements and checks of \eqref{eq:cw-split-check} can be
both done in $\mcl{O}(n^2)$ without any additional assumption about
$X$, $A$ and \mfk W nor the use of particular data structures or
primitives.  Therefore the asymptotic upper bound for time complexity
of the proposed algorithm is $\mcl{O}(nm(\mcl{L}_\mfk  W(n) + n^2))$
where $\mcl{L}_\mfk  W(n)$ is the upper bound for the complexity
introduced by computing the weak transitions for a given set of
states.

\begin{figure}[t]
\begin{algorithmic}[1]
\STATE $\mcl{X} \leftarrow \{X\}$
\STATE $\mcl{X}'\leftarrow \emptyset$
\REPEAT
  \STATE $changed \leftarrow\FALSE$
  \FORALL{$C \in \mcl{X}\setminus\mcl{X}'$}
    \FORALL{$\alpha \in A+\{\tau\}$}
		\STATE compute and sort $\rho(x,\hat\alpha,C)$ by block and weight
    \ENDFOR
    \IF{there is any split}
      \STATE $\mcl{X}'\leftarrow \mcl{X}$
      \STATE $\mcl{X} \leftarrow refine(\mcl X,C)$
      \STATE $changed \leftarrow\TRUE$
    \ENDIF
  \ENDFOR
\UNTIL{\NOT$changed$}
\RETURN \mcl X
\end{algorithmic}
\hrule{}
\caption{An alternative algorithm for linearly ordered blocks and weights.}
\label{fig:cw-algo-with-ords}
\end{figure}

The time complexity can be lowered by means of more efficient representations of
systems, partitions and weights. For instance, the structure of every semiring can
be used to define an ordering for its elements (\cf\cite{han2012:ordsemiring})
allowing the use of lookup data structures. 
Under the assumption of some linear ordering for weights and blocks 
(at least within the same partition) 
the operations of refinement and split testing can be carried out
more efficiently by sorting lexicographically the transitions ending in the
splitting block $C$. The resulting algorithm is reported in 
Figure~\ref{fig:cw-algo-with-ords}. 

This allows the algorithm to carry out the refinement of \mcl X
while it is reading the lexicographically ordered list of the saturated 
transitions. In fact, a block $B$ is split by $\langle\alpha,C\rangle$ if the list
contains different weights in the portion of the list where $B$ appears.
A change in the weights correspond to two states $x$ and $y$ such that 
\eqref{eq:cw-split-check} holds. 
For each $\langle\alpha,C\rangle$ there are at most $n$ weak transitions 
${\rho(x,\hat\alpha,C)}$ and these are sorted in
$\mcl O(n\mrm{ln}(n))$ -- or in $\mcl O(n)$ using a classical algorithm from 
\cite{aho74:algobook}. On the worst case the algorithm encounters $\mcl O(n)$ 
blocks during its entire execution yielding a worst case time complexity in 
$\mcl{O}(nm(\mcl{L}_\mfk  W(n) + n))$.

Overall, we have proved the following result:
\begin{proposition}
The asymptotic upper bound for time complexity of the algorithm is 
in $\mcl{O}(nm(\mcl{L}_\mfk  W(n) + n^2))$, for the general case, and in
$\mcl{O}(nm(\mcl{L}_\mfk  W(n) + n))$ given a linear ordering for blocks and weights.
Both algorithms have space complexity in $\mcl O(mn^2)$.
\end{proposition}

\section{Coalgebraic perspective}\label{sec:cat-view}
In this Section we illustrate the categorical construction behind
Definition~\ref{def:wlts-weak}.  The presentation is succinct due to
space constraints but it is based on general results from coalgebraic
theory.  In particular, we define weak bisimulations as cocongruences
of saturated or weak systems extending the elegant approach proposed
by Silva and Westerbaan in \cite{sw2013:epsilon}.

This is not the first work on a coalgebraic perspective of weak
bisimulations coalgebraically, as in the recent years there have been
several works in this direction.  In general, the approach is to
recover weak bisimulation as the coalgebraic bisimulation of saturated
systems.  In \cite{sokolova05:entcs,sokolova09:sacs} Sokolova et
al.~studied the case of action-based coalgebras and demonstrated their
results on the cases of non-deterministic and fully-probabilistic
systems.  In particular, the latter required to change the category of
coalgebras.  Recently, Brengos
\cite{brengos2012:ifip,brengos2013:corr} proposed an interesting
construction based on ordered-functors which yields saturated
coalgebras for the same behavioural functor.  Both these constructions
are parametric in the notion of saturation and are therefore way more
general; \cite{brengos2013:corr} describes an algebraic
structure and some conditions yielding precisely saturations for weak
bisimulations.  However, this approach does not cover the case of
generative and stochastic systems \cite[Sec.~6]{brengos2012:ifip} yet.
In \cite{ps2012:cmcs} Soboci\'{n}ski describes a neat account of weak
(bi)simula\-tion for non-deterministic systems and proves that
saturation via the double-arrow construction (\ie~$\tau$-closure)
results from a suitable change of base functor having a left adjunct
in the 2-categorical sense.

Likewise, we rely on saturation of the given systems but we do not
require any additional parameter.  Moreover, we base our definition on
cocongruences which allow us to work explicitly with the equivalence
classes and saturate the given coalgebras such that these describes
how each class is reached by each state without the need to alter the
behavioural functor.

Our saturation construction builds on the account of
$\epsilon$-transitions recently given in \cite{sw2013:epsilon} and on
the neat coalgebraic perspective of trace equivalence given by Hasuo
in \cite{hasuo2010:trace}.  Therefore the same settings are assumed,
\ie~we consider coalgebras for functors like $TF$ where $T$ and $F$
are endofunctors over a category \cat C with all finite limits and
$\omega$-colimits; $(T,\mu,\eta)$ is a monad; there exists a natural
transformation $\lambda$ distributing $F$ over $T$; the Kleisli
category $\cat{Kl}(T)$ is CPPO-enriched and has, for any $X \in \cat
C$, a final $(\hole + X)$-coalgebra.  Before describing the saturation
construction let us state the main definition of this Section.
\begin{definition}\label{def:coalg-weak}
Given two $TF_\tau$-coalgebras $(X,\alpha)$ and $(Y,\beta)$, a span of
jointly monic arrows $X \xleftarrow{p} R \xrightarrow{q} Y$ describes
a weak bisimulation between $\alpha$ and $\beta$ if and only if
there exists an epic cospan $X \xrightarrow{ f} C \xleftarrow{ g} Y$
such that $(R,p,q)$ is the final span to make the following diagram
commute:
\[
\begin{tikzpicture}[auto,xscale=2.3,yscale=1.2,font=\footnotesize,
	extended/.style={shorten >=-#1, shorten <=-#1},
	baseline=(current bounding box.center)]
		\node (n0) at (0,1) {\(X\)};
		\node (n1) at (2,1) {\(Y\)};
		\node (n2) at (1,.5) {\(C\)};
		\node (n3) at (0,0) {\(TF_\tau X\)};
		\node (n4) at (2,0) {\(TF_\tau Y\)};
		\node (n5) at (1,-.5) {\(TF_\tau C\)};
		\node (n6) at (-.7,1) {\(X\)};
		\node (n7) at (2.7,1) {\(Y\)};
		\node (n8) at (-.7,0) {\(TF_\tau X\)};
		\node (n9) at (2.7,0) {\(TF_\tau Y\)};
		\node (n10) at (1,1.5) {\(R\)};
		\draw[->] (n0) to node [swap] {\( f\)} (n2);
		\draw[->] (n1) to node [] {\( g\)} (n2);
		\draw[->] (n0) to node [swap] {\(\alpha^w\)} (n3);
		\draw[->] (n1) to node [] {\(\beta^w\)} (n4);
		\draw[->] (n2) to node [swap] {\(\gamma\)} (n5);
		\draw[->] (n3) to node [swap] {\(TF_\tau f\)} (n5);
		\draw[->] (n4) to node [] {\(TF_\tau g\)} (n5);
		\draw[->] (n6) to node [swap] {\(\alpha\)} (n8);
		\draw[->] (n7) to node [] {\(\beta\)} (n9);
		\draw[->] (n10) to node [swap] {\(p\)} (n0);
		\draw[->] (n10) to node [] {\(q\)} (n1);
\end{tikzpicture}
\]
where $\alpha^w$ and $\beta^w$ are the \emph{weak saturated} $TF_\tau$-coalgebras 
\wrt~$ f$ and $g$. 
\end{definition}

Let us see how the weak saturation $\alpha^w$ is defined.  In our
setting, the traces of a $TF$-coalgebra $\alpha$ are described by the
final map $\mrm{tr}_\alpha$ from the lifting of $\alpha$ in
$\cat{Kl}(T)$ to the final $\overline F$-coalgebra where $\overline F$
is the lifting of $F$ to $\cat{Kl}(T)$ induced by the distributive law
$\lambda : FT \Longrightarrow TF$ (\cf~\cite{hasuo2010:trace}).  Rawly
speaking, the monad $T$ can be thought as describing the branching of
the system whereas the observables are characterized by $F$. Assuming
this point of view, any $F$ can be extended with silent actions $\tau$
as the free pointed functor
\[
   F_\tau \defeq X + FX\text.
\]
Now, a $TF_\tau$-coalgebra $\alpha$ can be ``determinized'' by means
of its \emph{iterate} \cite{sw2013:epsilon}:
\[\mrm{itr}_\alpha \defeq \nabla_{FX}\circ\mrm{tr}_\alpha\]
where $\nabla$ is the codiagonal; the traces refer to $\alpha$ seen as
a $T(X + B)$ for $B = FX$ and $(X,\mrm{itr}_\alpha)$ is a
$TF$-coalgebra.  The iterate offers an elegant and general way to
``compress'' executions with leading silent transitions like $\tau^*a$
into single-step transitions with exactly one observable but retaining
the effects of the entire execution within the monad $T$.

These results can be used to cover executions ending with an observable
and hence do not directly lend themselves to equivalences based also on
trailing silent actions like in the case $\tau^*a\tau^*$, as required by the
weak bisimulation. However, let us suppose to have, for any given $TF_\tau$-coalgebra
$(X,\alpha)$, the $T$-coalgebra $(X,\alpha^\tau)$ describing how each state
reaches every class with $\tau$-transition only; then, the coalgebra 
describing reachability by $\tau^*a\tau^*$ is exactly:
\[
\alpha^\flat :
X \xrightarrow{\mrm{itr}_\alpha} TFX \xrightarrow{TF{\alpha^\tau}}
TFTX \xrightarrow{\lambda_X} TTFX \xrightarrow{\mu} TFX
\text.
\]
Then the saturated coalgebra $\alpha^w$ is defined by means of the
2-cell structure of $\cat{Kl}(T)$ as the join described by the diagram below.
\begin{equation}
\label{eq:coalg-split}
\begin{tikzpicture}[auto,xscale=2.5,yscale=1.2,font=\footnotesize,rotate=-45,
	baseline=(current bounding box.center)]
		\node (n0) at (0,1) {\(X\)};
		\node (n1) at (0,0) {\(X\)};
		\node (n2) at (1,1) {\(FX\)};
		\node (n3) at (1,0) {\((X+F X)\)};
		\draw[->] (n0) to node [swap] {\(\alpha^\tau\)} (n1);
		\draw[->] (n0) to node [] {\(\alpha^\flat\)} (n2);
		\draw[->] (n1) to node [swap] {\(\iota_1\)} (n3);
		\draw[->] (n2) to node [] {\(\iota_2\)} (n3);
		\draw[->] (n0) -- (n3);
		\node[preaction={fill=white,inner sep = 0pt}] at ($(n1)!.5!(n2)$) {\(\alpha^w\defeq\sqcup\)};
		\node[] at ($(n1)!.25!(n2)$) {\(\sqsubseteq\)};
		\node[] at ($(n2)!.25!(n1)$) {\(\sqsupseteq\)};
\end{tikzpicture}\end{equation}
This definition points out that $\tau^*$ and
$\tau^*a\tau^*$ are two close but different cases.

In order to define the $T$-coalgebra $\alpha^\tau$, first we
need to be able to consider only the silent action of the given
$\alpha$. This information can be isolated from $\alpha$
by means of the same structure used in \eqref{eq:coalg-split}.
Therefore we define, for every $TF_\tau$-coalgebra $\alpha$,
its silent and observable parts, namely $\alpha^s$ and
$\alpha^o$, as the (greatest) arrows to make the following 
diagram commute and have $\alpha$ as their join.
\[
\begin{tikzpicture}[auto,xscale=2.5,yscale=1.2,font=\footnotesize,rotate=-45,
	baseline=(current bounding box.center)]
		\node (n0) at (0,1) {\(X\)};
		\node (n1) at (0,0) {\(X\)};
		\node (n2) at (1,1) {\(FX\)};
		\node (n3) at (1,0) {\((X+F X)\)};
		\draw[->] (n0) to node [swap] {\(\alpha^s\)} (n1);
		\draw[->] (n0) to node [] {\(\alpha^o\)} (n2);
		\draw[->] (n1) to node [swap] {\(\iota_1\)} (n3);
		\draw[->] (n2) to node [] {\(\iota_2\)} (n3);
		\draw[->] (n0) -- (n3);
		\node[preaction={fill=white,inner sep = 0pt}] at ($(n1)!.5!(n2)$) {\(\alpha=\sqcup\)};
		\node[] at ($(n1)!.25!(n2)$) {\(\sqsubseteq\)};
		\node[] at ($(n2)!.25!(n1)$) {\(\sqsupseteq\)};
\end{tikzpicture}\]

Because $\alpha^\tau$ has to describe how each class is reached,
classes can be used as the observables needed to apply the iterate
construction to $\alpha^\tau$. However, to be able to select the class
to be reached and consider it as the only one observable by the
iterate (likewise $F_\tau$ distinguish silent and observable actions
by means of a coproduct) we need $X$ and $C$ to be represented as
indexed coproducts of simpler canonical subobjects corresponding to
the classes induced by $f : X \to C$.
Henceforth, for simplicity we assume $X \cong X\cdot 1$ and $C \cong
C\cdot 1$. %
For each class $c : 1 \to C$ let $X\cong \overline{X}_c + X_c$ be the
split induced by $c$.  This extends to the coalgebra $\alpha^s$ (by
coproduct) determining the coalgebra: $\alpha^s_c : \overline{X}_c \to
T(\overline{X}_c + X_c)$ whose iterate is the map $\alpha^+_c :
\overline{X}_c \to T(X_c)$ describing executions of silent action only
ending in $c$ (but starting elsewhere). This yields a $C$-indexed
family of morphisms which together describe $\tau^+$ and the
information is collected in one $T$-algebra as a join in the 2-cell
like \eqref{eq:coalg-split}.  For this join to be admissible we
require $T$ to not exceed the completeness of 2-cells, \ie~for any $x
: 1 \to X \in \cat C$ the supremum of the set $\cat{Kl}(T)(1,X)$
determined by $\cat{Kl}(T)(x,X)$ exists. Reworded if cells are
$\kappa$-CPPOs, then $T$ is $\kappa$-finitary; \eg~in \cat{Set}
$T$-coalgebras describe image $\kappa$-bounded $T$-branching
systems. Thus, for every $x$ the family of arrows $\{\alpha^+_c\circ
x\}$ is limited by $\kappa$ and can be joined. These are 
composed in $\alpha^+$ as the universal arrow in the $X$-fold
coproduct.  The last step is provided by the monad
unit which is a $T$-coalgebra describing how states reach their
containing class and can be easily joined to the above obtaining,
finally, $\alpha^\tau$. This completes the construction of $\alpha^w$.

\paragraph{Weighted labelled transition systems}
Assuming at most countably many actions, image-countable $\mfk W$-WLTs 
are in 1-1 correspondence with $\fw(A\times\hole)$-coalgebras where
$(\fw,\mu,\eta)$ is the monad of $\mfk W$-valued functions with at most
countable support. In particular, $\fw X$ is the set morphisms from
$\cat{Set}(X,W)$ factoring through $\mbb N$. On arrows $\fw$ is defined as 
$\fw f (\varphi)(y)\defeq \sum_{x\in f^{-1}(y)}\varphi(x)$.  The unit $\eta$
is defined as $\eta(x)(y) \defeq 1$ for $x = y$ and $0$ otherwise, and the 
multiplication $\mu$ is defined as 
$\mu(\psi)(x) \defeq \sum_{\varphi}\psi(\varphi)\cdot\varphi(x)$.

If $\mfk W$ is the boolean semiring, $\fw$ is precisely the countable
powerset monad.  Strength and double strength readily generalize to
every $\fw$ and by \cite{hasuo07:trace} there is a canonical law
$\lambda$ distributing $(A\times\hole)$ over $\fw$.  The semiring
$\mfk W$ can be easily endowed with an ordering which lifts point-wise
to $\mfk W$-valued functions \cite{han2012:ordsemiring}. In
particular, any $\omega$-semiring with a natural order 
(\cf~Section~\ref{sec:good-semirings}) yields a 
CPPO-enriched $\cat{Kl}(\fw)$ with bottom the constantly $0$ function.

\begin{theorem}
Let $T$ be $\fw$ and $F$ and $(A\times\hole)$. For any given 
$TF_\tau$-coalgebra $\alpha$ and its corresponding $\mfk W$-LTS,
Definition~\ref{def:coalg-weak} and Definition~\ref{def:wlts-weak}
coincide.
\end{theorem}
\begin{proof}
By unfolding of Definition~\ref{def:coalg-weak} and by minimality of
executions considered by the construction of $\alpha^w$.
\end{proof}

\section{Conclusions and future work}
\label{sec:concl}

In this paper we have introduced a general notion of \emph{weak
  weighted bisimulation} which applies to any system that can be
specified as a LTS weighted over a \emph{semiring}.  The semiring
structure allows us to compositionally extend weights to multi-step
transitions.  
We have shown that our notion of weak
bisimulations naturally covers the cases of non-deterministic, 
fully-probabilistic, and stochastic systems, among others.  
We described a ``universal'' algorithm for computing weak bisimulations
parametric in the underlying semiring structure and proved its
decidability for every positively ordered $\omega$-semiring.
Finally, we gave a categorical account of the coalgebraic construction behind
these results, providing the basis for extending the results presented
here to other behavioural equivalences. 

Our results came with a great flexibility offered, from one hand,
by the possibility to instantiate WLTSs to several systems (by 
just providing opportune semirings) and, from the other,
by the possibility to consider many other behavioural equivalences
simply by changing the observation patterns used 
in Definition~\ref{def:wlts-weak} and in the linear 
equations systems at the core of the proposed algorithm
as described by Example~\ref{ex:delay-bisim}.

A possible future work is to improve the efficiency of our algorithm,
\eg~by extending Paige-Tarjan's algorithm for strong bisimulation
instead of Kannellakis-Smolka's, or using more recent approaches
based on symbolic bisimulations \cite{wimmer2006:sigref}.

The algorithm presented is based on
Kanellakis-Smolka's. A possible future work could be to improve the
efficiency of this algorithm, \eg~by extending Paige-Tarjan's
algorithm for strong bisimulation, or more recent approaches like
symbolic bisimulations (\eg~\cite{wimmer2006:sigref}).

Obviously, for specific systems and semirings there are solutions
more efficient than our.  For instance, in the case of systems over the semiring of
non-negative real numbers (which captures \eg~probabilistic
systems) the asymptotic upper bound for time complexity of our
algorithm is $\mcl O(mn^{3.8})$ (since $\mcl{L}_{\mbb R^+_0}(n)$ is in
$\mcl O(n^{2.8})$ using \cite{aho74:algobook}).  However, deciding
weak bisimulation for fully-probabilistic systems is in $\mcl O(mn^3)$
on the worst case using the algorithm proposed by Baier and Hermanns
in \cite{baier00} (the original analysis assumed $A$ to be fixed
resulting in the worst case complexity $\mcl O(n^3)$).  Their
algorithm capitalise on properties not available at the general level of
WLTSs (even under the assumption of suitable orderings), such as:
sums of outgoing transitions are bounded, 
there are complementary events, real numbers have more structure
than a semiring, weak and delay bisimulations coincide for finite 
fully-probabilistic systems (\eg~this does not hold for non-deterministic LTSs). The aim of
future work is to generalize the efficient results of
\cite{baier00} to a parametrized algorithm for constrained WLTSs, or
at least for some classes of WLTSs subject to suitable families of
constraints. 

The construction presented in Section~\ref{sec:cat-view}
introduces some techniques and tools that can be used to deal with
other behavioural equivalences.  In fact, we think that many
behavioural equivalences can be obtained by ``assembling'' smaller
components, by means of 2-splits, 2-merges and \emph{iterate}, as we
did for weak bisimulation.  We plan to provide a formal, and easy to
use, language for describing and combining these ``building blocks''
in a modular way.

An important direction for future work is to generalize our
framework by weakening the assumptions on the underlying category
(introduced in order to observe and manipulate equivalence classes)
and by considering different behavioural functors.  In particular, we
intend to extend this framework to \textsc{ULTraS}s, \ie~the generalization of
WLTSs recently proposed by Bernardo et al.~in
\cite{denicola13:ultras}. These are an example of \emph{staged transition systems}, where
several behavioural functors (or stages) are ``stacked'' together.

\bibliographystyle{abbrv}
\bibliography{allbib}

\end{document}